\documentclass[conference]{IEEEtran}

\usepackage[cmex10]{amsmath}
\usepackage{braket}

\usepackage{amsfonts}

\usepackage{xcolor}

\usepackage[colorlinks=false,frenchlinks=false,pdfborder={0 0 0},naturalnames=false,hypertexnames=false]{hyperref}
\newcommand{\email}[1]{\protect\href{mailto:#1}{#1}}
\colorlet{mylinkcolor}{green!50!black}
\hypersetup{
  colorlinks = true,
  allcolors = mylinkcolor,
}

\usepackage{algorithm}
\usepackage{algpseudocode}

\usepackage[amsmath,thmmarks,hyperref]{ntheorem}

\usepackage[capitalize,nameinlink]{cleveref}
\crefformat{equation}{#2(#1)#3}
\crefrangeformat{equation}{(#3#1#4--#5#2#6)}
\crefmultiformat{equation}{#2(#1)#3}{ and #2(#1)#3}
{, #2(#1)#3}{, and #2(#1)#3}
\crefrangemultiformat{equation}{(#3#1#4--#5#2#6)}%
{ and (#3#1#4--#5#2#6)}{, (#3#1#4--#5#2#6)}{, and (#3#1#4--#5#2#6)}
\Crefformat{equation}{#2Equation~(#1)#3}
\Crefrangeformat{equation}{Equations~(#3#1#4--#5#2#6)}
\Crefmultiformat{equation}{Equations~(#2#1#3)}{ and (#2#1#3)}
{, (#2#1#3)}{, and (#2#1#3)}
\Crefrangemultiformat{equation}{Equations~(#3#1#4--#5#2#6)}%
{ and (#3#1#4--#5#2#6)}{, (#3#1#4--#5#2#6)}{, and (#3#1#4--#5#2#6)}

\newcommand{\proofbox}{\vbox{\hrule height0.6pt\hbox{\vrule height1.3ex width0.6pt\hskip0.8ex\vrule width0.6pt}\hrule height0.6pt}}
\theoremstyle{nonumberplain}
\theoremheaderfont{\normalfont\itshape}
\theorembodyfont{\normalfont}
\theoremseparator{.}
\theoremsymbol{\proofbox}
\newtheorem{proof}{Proof}
\theoremstyle{plain}
\theoremheaderfont{\normalfont\sc}
\theorembodyfont{\normalfont\itshape}
\theoremseparator{.}
\theoremsymbol{}
\newtheorem{theorem}{Theorem}
\newtheorem{lemma}[theorem]{Lemma}

\theoremstyle{plain}
\theoremheaderfont{\normalfont\sc}
\theorembodyfont{\normalfont}
\theoremseparator{.}
\theoremsymbol{}
\newtheorem{scenario}{Scenario}
\newtheorem{definition}{Defn.}
\theoremstyle{plain}
\theoremheaderfont{\normalfont\itshape}
\theorembodyfont{\normalfont}
\theoremseparator{.}
\theoremsymbol{}

\crefname{claim}{Claim}{Claims}

\usepackage{tikz}
\usetikzlibrary{decorations,fit}
\tikzset{%
  gnode/.style={shape=circle,minimum size=3mm,fill,draw=black}
}

\newcommand{\TpgEdges}{

  \node (i) at (0,1.5) {i};
  \node (j) at (2,1.75) {j};

  \foreach \k in {2,...,5}
  \draw (i) -- (1,0.5*\k-0.5);

  \foreach \k in {3,...,6}
  \draw (j) -- (1,0.5*\k-0.5);
}

\newcommand{\TpgNodes}{

  \foreach \i in {1,...,4}
  \node [color=red!25,gnode] at (0,0.5*\i) {.};

  \foreach \k in {1,...,6}
  \node [color=blue!25,gnode] at (1,0.5*\k-0.5) {.};

  \foreach \j in {1,...,4}
  \node [color=green!25,gnode] at (2,0.5*\j-0.25) {.};

  \node at (i) {$i$};
  \node at (j) {$j$};
}

\usepackage{graphicx,epstopdf}

\usepackage[font=footnotesize,justification=centering]{subfig}

\newcommand{\Real}{\mathbb{R}}
\newcommand{\Note}[1]{\textcolor{red}{#1}}

\newcommand{\qtext}[1]{\quad\text{#1}\quad}
\newcommand{\MR}[2]{\lowercase{#1}_{#2*}} % Matrix row
\newcommand{\MC}[2]{\lowercase{#1}_{*#2}} % Matrix column
\DeclareMathOperator{\sgn}{sgn}
\newcommand{\Exp}[1]{\mathbb{E}[#1]}
\DeclareMathOperator{\Prob}{Pr}
\newcommand{\nnz}[1]{\text{nnz}(#1)}
\DeclareMathOperator{\Deg}{deg}

\newcommand{\eps}{\varepsilon}
\newcommand{\bound}{K}
\newcommand{\dotprod}[2]{\vec{#1} \cdot \vec{#2}}
\newcommand{\storage}[1]{\text{storage}(#1)}
\newcommand{\hideme}[1]{}

\begin{document}

\title{Diamond Sampling for Approximate\\ Maximum All-pairs Dot-product (MAD) Search}

 \author{%
 \IEEEauthorblockN{Grey Ballard, Tamara G. Kolda, and Ali Pinar}%
 \IEEEauthorblockA{Data Sciences \& Cyber Analytics Department\\%
 Sandia National Laboratories\\Livermore, CA\\%
 \email{gmballa@sandia.gov}, \email{tgkolda@sandia.gov}, \email{apinar@sandia.gov}}%
 \and
 \IEEEauthorblockN{C. Seshadhri}%
 \IEEEauthorblockA{Department of Computer Science\\%
 University of California\\Santa Cruz, CA\\\email{scomandu@ucsc.edu}}%
 }

\maketitle

\begin{abstract}\boldmath
  Given two sets of vectors, $A = \set{\vec{a_1}, \dots, \vec{a_m}}$ and
  $B=\set{\vec{b_1},\dots,\vec{b_n}}$, our problem is to find the top-$t$ dot
  products, i.e., the largest $|\vec{a_i}\cdot\vec{b_j}|$ among all possible
  pairs. This is a fundamental mathematical problem that appears in numerous data
  applications involving similarity search, link prediction, and collaborative filtering.
  We propose a sampling-based approach that avoids direct computation of all $mn$ dot products.
  We select diamonds (i.e.,
  four-cycles) from the weighted tripartite representation of $A$ and
  $B$. The probability of selecting a diamond corresponding to
  pair $(i,j)$ is proportional to $(\vec{a_i}\cdot\vec{b_j})^2$, amplifying the focus on the
  largest-magnitude entries.
  Experimental  results indicate that diamond sampling is orders of magnitude faster than
  direct computation and
  requires far fewer samples than any competing approach. 
  We also apply diamond sampling to the special case of maximum inner product search, and get
  significantly better results than the state-of-the-art hashing methods. 
\end{abstract}

\section{Introduction}
\label{sec:intro}

Finding similar items is a fundamental problem that underlies numerous 
problems in data analysis. Link prediction in a graph can be cast as finding
similar nodes in the graph~\cite{AdAd03,LiKl07}; customers are recommended similar products~\cite{CrKoTu10}; text
analysis often involves finding similar texts~\cite{SaAlBu93,BeDuOb95}; data cleaning
requires removal of entries that are essentially identical~\cite{KaMeCh05}. In these settings,
entities are represented as vectors in high-dimensional feature space, i.e.,
$\vec{v} \in \Real^d$, for some large $d$. Many 
notions of similarity involve dot products, so a measure of distance between
$\vec{v}$ and $\vec{w}$ is $\vec{v}\cdot \vec{w}$. This subsumes
cosine similarity, common-neighbors, database join operations \cite{AnPi05},
frequent itemset mining, data cleaning \cite{KaMeCh05}, etc.
Motivated by these applications,
we study the Maximum All-pairs Dot-product (MAD) problem.

\begin{definition}[$t$-MAD: Max All-pairs Dot-product]%
Given two sets of $d$-dimensional vectors $A = \{\vec{a_1},\ldots,\vec{a_m}\}$
and $B = \{\vec{b_1},\ldots,\vec{b_n}\}$: find the index pair $(i,j)$
that maximizes $|\vec{a_i} \cdot \vec{b_j}|$. More generally, given
additional parameter $t$, find the $t$ index
pairs $\{(i_1, j_1), \ldots, (i_t, j_t)\}$ corresponding
to the $t$ largest dot products.
\end{definition}

It is convenient to think of $A$ and $B$ as matrices ($A \in \Real^{d \times m}$
and $B \in \Real^{d \times n}$), where the columns are the corresponding
vectors. The MAD problem is exactly finding the largest entries
in the product $C = A^TB \in \Real^{m \times n}$. 

The MAD formulation subsumes many existing problems in the literature.
For the special case where $A$ is a single column (equivalently, $m=1$), this is the exactly the MIPS (Maximum Inner
Product Search) problem~\cite{CoLe99,ShLi14}. Here, we maximize the dot product with $\vec{a}$
among all columns of $B$. When $A = B$ is the adjacency matrix of a graph,
this is equivalent to finding pairs of nodes with the most common neighbors,
a fundamental link prediction operation.

\begin{figure}
\centering
\includegraphics[height=2.5in]{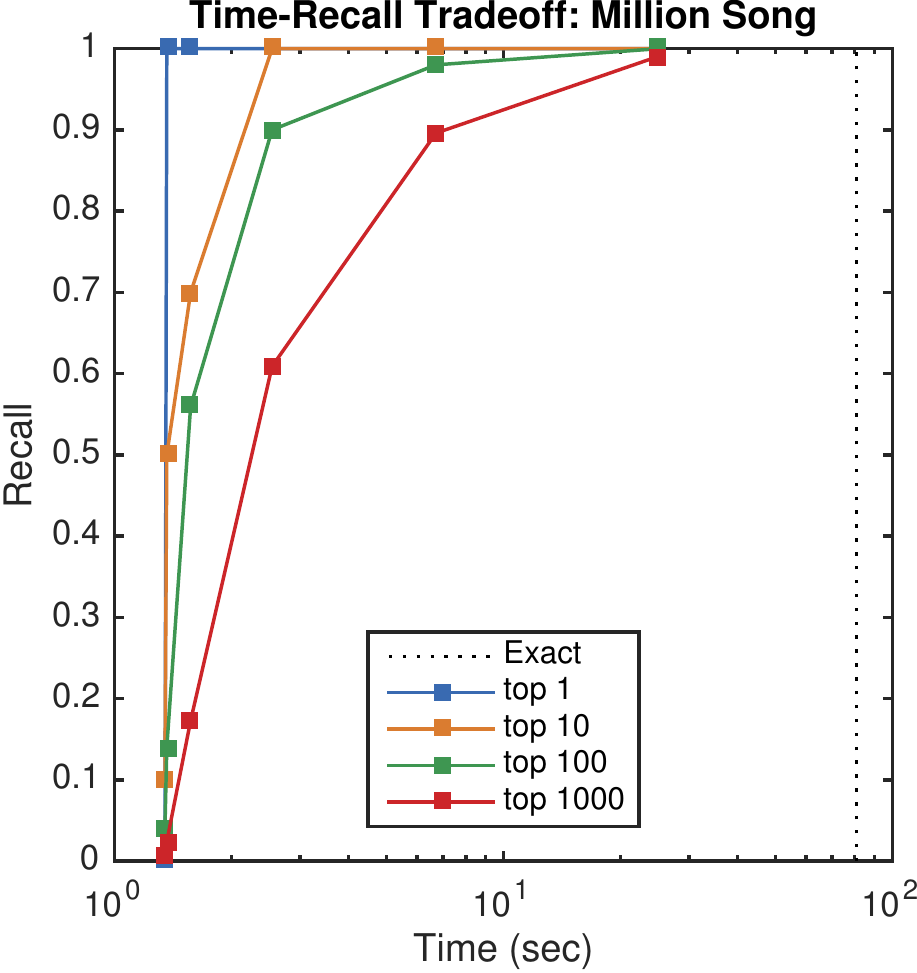}
\caption{\small Million song dataset~\cite{BeElWhLa11,McBeElLa12,echonest}: Dataset with 350K songs, and 48M user-song
entries. Diamond sampling finds top 100 correlated songs within 7 seconds, while exhaustive search takes around
80 seconds.}
\label{fig:awesome}
\end{figure}

\subsection{Difficulties with previous art} \label{sec:chal}

The most obvious approach is to simply compute $A^TB$ exhaustively. There is a rich history
on algorithms for dense and sparse matrix multiplication~\cite{Gu78,AmPa09}, with
implementations in libraries like Intel MKL's BLAS~\cite{MKL14} and CSparse~\cite{Davis06}. However,  existing algorithms become prohibitive as the sizes of $A$ and $B$ grow, even if they are sparse.

There is also much literature in approximate matrix multiplication,
which quickly computes an approximate product $\tilde C \approx A^TB$.
Drineas, Kannan, and Mahoney \cite{DrKa01,DrKaMa06} introduced an approach based on
sampling rows of $A$ and $B$ to minimize Frobenius norm of the error. 
Much study has been done on various sampling strategies~\cite{DrMa05,Sa06,BeWo08,MaZo11,Pa13,HoIp15}.
These methods are also not suited for MAD, since we only care for a few entries (at most, say, 1000)
despite the matrix having dimensions in the millions.

An alternate, popular approach for high-dimensional nearest neighbor search is
some form of dimension reduction, the most famous approach being Indyk and Motwani's Locality Sensitive Hashing
(LSH)~\cite{InMo98,GiInMo99,AnIn08}. Recent results by Shrivastava and Li extend
LSH methods for the MIPS problem~\cite{ShLi14}.
These approaches usually involve randomly hashing the vectors
into a few ``buckets", and then searching for vector pairs that share
many common buckets.
For
high-dimensional data, the maximum dot product is often small in
comparison to the vector norms, so
the similarities are quite small.  
This
means that many hashes are required to find the nearest neighbors,
leading to a storage blowup (typically, two orders of magnitude more than the data). This
can be quite prohibitive even for small data sets.

\subsection{The sampling approach}

We take a different route, and apply \emph{index sampling} methods. The idea is to sample pair $(i,j)$
proportional to some function of the dot product $\vec{a_i}\cdot \vec{b_j}$. With enough samples, we hope to find the
large entries of $C = A^TB$. The earliest application of this idea
is by Cohen and Lewis, who constructed a sampling algorithm for
the MIPS problem~\cite{CoLe97,CoLe99}. 
Their approach samples pairs $(i,j)$ 
proportional to $\vec{a_i} \cdot \vec{b_j}$~\cite{CoLe97,CoLe99}.
Campagna and Pagh give sampling approaches for a variety of distance measures~\cite{CaPa11}.

We stress that no previous result (either sampling based, LSH based, or otherwise) applies directly
to the MAD problem.

\subsection{Our Contributions} \label{sec:result}

{\bf Diamond sampling:} Our main contribution is \emph{diamond sampling}, a new randomized approach to the MAD problem.
This is inspired by recent work by Jha et al for 4-vertex motif detection in large graphs~\cite{JhSePi15}. Their 
idea is to sample 3-paths in graphs to estimate counts of 4-vertex motifs. We generalize that idea
to the matrix product setting, to design a sampling procedure for the MAD problem.
Diamond sampling is able to sample pairs $(i,j)$ proportional to $(\vec{a_i} \cdot \vec{b_j})^2$.
The square term is a critical improvement over Cohen and Lewis; it allows for faster convergence
to the top entries of $C = A^TB$.

{\bf Theoretical analysis:} We give a theoretical analysis of diamond sampling,
and prove concentration bounds on its behavior. Our analysis shows the eventual convergence of 
the sampling to squared entries, with no assumption whatsoever. Previous sampling work
required nonnegativity assumption, or assumed structural correlations among positive entries~\cite{CoLe97,CoLe99}.
We give strong storage bounds on diamond sampling, and show that it requires very little overhead.

{\bf Empirical validation:} 
We apply diamond
sampling on six real-world datasets and show that is extremely efficient. Diamond sampling is orders
of magnitude faster than exact computation and requires far fewer samples compared to other
matrix sampling approaches. \cref{fig:awesome} shows the results of diamond sampling
on a song dataset with 48M user-song entries. We consider $A = B$ to be the matrix
where songs are columns and attempt to find the top correlated songs. We can get the top
100 pairs in an order of magnitude less time than exhaustive computation.

We consider numerous applications in product recommendation and link prediction, and consistently
get to top 10-100 dot products, with a speedup of 10-100X over exact computation. Furthermore,
the number of samples required is much smaller than the Cohen-Lewis approach~\cite{CoLe99}.

{\bf Application to MIPS:} Given recent interest in the MIPS problem, we also apply diamond sampling to a MIPS problem, as used in~\cite{ShLi14}.
We focus on the MovieLens dataset~\cite{movielens}, a collaborative filtering application.
We get significantly better precision-recall curves with a maximum precision of 90-100\% as opposed
to 30-65\% with asymmetric LSH methods. Our running time is a fraction of a second per query. Our diamond sampling requires minimal storage overhead and  much less than LSH methods, which 
require large amounts of memory, easily running into hundred times the dataset
size (requiring a large-memory machine).

\section{Preliminaries and notation}
We use the standard notation that $[n]=\set{1,\dots,n}$.
For $x \in \Real$, the function $\sgn(x) = 1$ if $x \geq 0$ and $-1$
otherwise; i.e., $x = \sgn(x)|x|$.
Let $\vec{v} \in \Real^n$ denote a vector and $M = (m_{ij}) \in
\Real^{m \times n}$ be a matrix.  We denote the vector and
matrix $p$-norms as follows:
\begin{equation}\label{eq:pnorm}
   \|\vec{v}\|_p^p = \sum_{i=1}^n |v_i|^p \qtext{and} 
   \|M\|_p^p = \sum_{i=1}^m \sum_{j=1}^n |m_{ij}|^p.
\end{equation}
Note that the matrix $p$-norms are entrywise rather than the induced norms. 
We let $\nnz{M}$ denote the number of nonzeros in $M$; if $M$
is dense, then we suppose without loss of generality that $\nnz{M} = mn$.

We assume throughout that $A \in \Real^{d \times m}$ and $B \in
\Real^{d \times n}$.
We use $k,\, k' \in [d]$ to index rows of $A$ and $B$. The $k$th rows of
$A$ and $B$ (transposed to column vectors) are denoted by $\MR{A}{k}$
and $\MR{B}{k}$ respectively.
We use $i,\, i' \in [m]$ to index columns of $A$, whose $i$th column is
denoted by $\vec{a_i}$ or $\MC{A}{i}$.
We use $j,\, j' \in [n]$ to index columns of $B$, whose $j$th column is
denoted by $\vec{b_j}$ or $\MC{B}{j}$.
Since, by definition, $C=A^TB$, we have
\begin{displaymath}
  c_{ij} = \vec{a_i} \cdot \vec{b_j} = \sum_{k} a_{ki} b_{kj}
  \qtext{for all} i \in [m],\;  j \in [n].
\end{displaymath}

If $A$ and $B$ are binary, i.e., unweighted adjacency matrices, then
we can consider them as representing a tripartite graph on $m
+ d + n$ nodes; see \cref{fig:tripartite}. From this interpretation, we 
define the neighbor sets, 
\begin{align*}
  \mathcal{N}^A_i & = \set{ k \in [d] | a_{ki} = 1}, &
  \mathcal{N}^A_k & = \set{ i \in [m] | a_{ki} = 1}, \\
  \mathcal{N}^B_j & = \set{ k \in [d] | b_{kj} = 1}, &
  \mathcal{N}^B_k & = \set{ j \in [n] | b_{kj} = 1}.
\end{align*}
Correspondingly, we can define  degrees of the nodes, i.e., 
\begin{align*}
  \Deg_i^A = | \mathcal{N}^A_i | = \| \MC{a}{i} \|_1, &
  \Deg_k^A = | \mathcal{N}^A_k | = \| \MR{a}{k} \|_1, \\
  \Deg_j^B = | \mathcal{N}^B_j | = \| \MC{b}{j} \|_1, &
  \Deg_k^B = | \mathcal{N}^B_k | = \| \MR{b}{k} \|_1.
\end{align*}
Random selection is uniform, i.e., equal probability for all elements
of a discrete set, unless stated otherwise.

\section{Diamond sampling}
\label{sec:diamond}

Complexity of diamond sampling depends on the two input matrices. We start our discussion with the special case of binary matrices $A$ and $B$. 
We follow with the general case and then discuss  other special cases such as nonnegative inputs and computing the maximum in $A^TA$, i.e., $B=A$.

\subsection{Binary inputs}
To motivate our procedure, we start with the case where $A$ and $B$ are binary
matrices.  We can represent this as a tripartite graph where the $m$
columns of $A$, indexed by $i$, correspond to nodes on the left; the
$n$ columns of $B$, indexed by $j$, correspond to nodes on the right;
and the $d$ common rows of $A$ and $B$, indexed by $k$ or $k'$,
correspond to nodes in the center.  Edge $(i,k)$ exists iff $a_{ki} =
1$; likewise, edge $(k,j)$ exists iff $b_{kj} = 1$.
Therefore, $c_{ij}$ is simply the number of common neighbors of $i$ and $j$:
\begin{displaymath}
  c_{ij} = \dotprod{a_i}{b_j} = |\set{k | k \in \mathcal{N}^A_i \cap \mathcal{N}^B_j}|.
\end{displaymath}

\begin{figure}[htbp]
  \centering
  \subfloat[Wedge $(i,k,j)$]{\label{fig:tpg-wedge}
  \begin{tikzpicture}
    \TpgEdges
    \node (k) at (1,0.5*3-0.5) {};
    \draw [very thick] (k)--(i);
    \draw [very thick] (k)--(j);
    \TpgNodes
    \node at (k) {$k$};
    \node [shape=circle,draw=black,minimum size=4mm,very thick] at (i) {};
    \node [shape=circle,draw=black,minimum size=4mm,very thick] at (j) {};
    \node [shape=circle,draw=black,minimum size=4mm,very thick] at (k) {};
  \end{tikzpicture}}
~
  \subfloat[Three-path $(k',i,k,j)$]{\label{fig:tpg-threepath}
  \begin{tikzpicture}
    \TpgEdges
    \node (kp) at (1,0.5*2-0.5) {};
    \node (k) at (1,0.5*5-0.5) {};
    \draw [very thick] (k)--(i);
    \draw [very thick] (k)--(j);
    \draw [very thick] (i)--(kp);
    \TpgNodes
    \node at (k) {$k$};
    \node at (kp) {$k'$};
    \node [shape=circle,draw=black,minimum size=4mm,very thick] at (i) {};
    \node [shape=circle,draw=black,minimum size=4mm,very thick] at (j) {};
    \node [shape=circle,draw=black,minimum size=4mm,very thick] at (k) {};
    \node [shape=circle,draw=black,minimum size=4mm,very thick] at (kp) {};
  \end{tikzpicture}}
~
  \subfloat[Diamond $(k',i,k,j)$]{\label{fig:tpg-diamond}
  \begin{tikzpicture}
    \TpgEdges
    \node (k) at (1,0.5*3-0.5) {};
    \node (kp) at (1,0.5*5-0.5) {};
    \draw [very thick] (k)--(i);
    \draw [very thick] (k)--(j);
    \draw [very thick] (kp)--(i);
    \draw [very thick] (kp)--(j);
    \TpgNodes
    \node at (k) {$k$};
    \node at (kp) {$k'$};
    \node [shape=circle,draw=black,minimum size=4mm,very thick] at (i) {};
    \node [shape=circle,draw=black,minimum size=4mm,very thick] at (j) {};
    \node [shape=circle,draw=black,minimum size=4mm,very thick] at (k) {};
    \node [shape=circle,draw=black,minimum size=4mm,very thick] at (kp) {};
  \end{tikzpicture}}
\caption{\small Illustration of tripartite graph. For simplicity, we show
  only those edges incident nodes $i$ or $j$.}
  \label{fig:tripartite}
\end{figure}
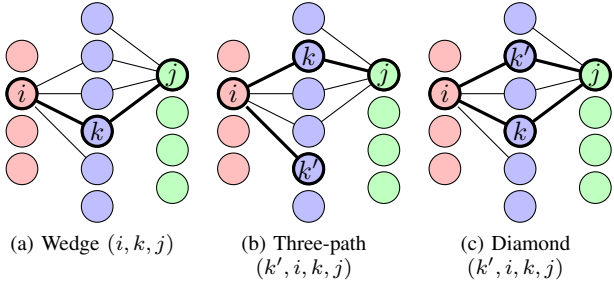

If node $k$ has an $A$-neighbor $i$ and a $B$-neighbor $j$, then we
call $(i,k,j)$ a ``wedge.''
The existence of such a wedge implies that $c_{ij} \geq 1$. In fact, there are exactly
$c_{ij}$ distinct wedges connecting pair $(i,j)$; see \cref{fig:tpg-wedge}.
The probability of selecting a 
random wedge with endpoints $(i,j)$ can be shown to be proportional to
$c_{ij}$ \cite{CoLe97,CoLe99}.

In diamond sampling, our goal is find a ``diamond'' $(k',i,k,j)$
formed by \emph{two} intersecting wedges, i.e., $(i,k,j)$ and
$(i,k',j)$; see \cref{fig:tpg-diamond}. Note that any pair $(i,j)$
participates in $c_{ij}^2$ diamonds (note that we are not requiring $k$ and $k'$ to be different). Hence, the probability of
selecting a random diamond of the form $(k',i,k,j)$ is proportional to
$c_{ij}^2$.

Sampling random diamonds will expedite  identifying  the largest
dot products as compared to sampling random wedges; however, sampling random
diamonds is more complex.  Thankfully, we can adapt the arguments of
Jha et al.~\cite{JhSePi15} for this purpose. Here, the goal is to find a random
three-path of the form $(k',i,k,j)$. If it closes to form a
four-cycle, then it is a random diamond.  Moreover, these samples will be uncorrelated.  That is, given a set of random 3-paths,  those that complete to a diamond will form a uniform sample of the  diamonds. 
See \cref{fig:tpg-diamond}
for a three-path that closes to form a diamond and
\cref{fig:tpg-threepath} for one that does not.

Finding a random three-path of the
form $(k',i,k,j)$ is a multi-step procedure, shown in
\cref{alg:diamond-binary} and illustrated in \cref{fig:diamond}.
In \cref{line:dbw}, we weight each edge $(i,k)$ according
to the number of three paths it is the center of, i.e., $\Deg^A_i
\Deg^B_k$ (again we do not require $k\neq k'$), and store the weights in a matrix $W$.
Observe that $W$ has the same sparsity pattern as $A$.
In \cref{line:dbs1}, we select a random
edge $(i,k)$ proportional to its weight (see \cref{fig:step1}).
To complete the three-path, we select a random neighbor of $k$ in
$B$, labeled $j$ in \cref{line:dbs2} (see \cref{fig:step2}) and a random neighbor of $i$
in $A$, labeled $k'$ in \cref{line:dbs3} (see \cref{fig:step3}).  This yields a uniform random three-path. If edge
$(k',j)$ exists, i.e., $b_{k'j} = 1$, then the three-path is a diamond
and so we increment the counter
$x_{ij}$ in \cref{line:dbs4}; obviously, $\nnz{X} \leq s$. 

\begin{algorithm}
  Given matrices $A \in \set{0,1}^{m \times d}$ and $B \in
  \set{0,1}^{n \times d}$. \\
  Let $s$ be the number of samples.
  \begin{algorithmic}[1]
    \For{$(k,i) \in [d] \otimes [m]$}
    %\ForAll{$a_{ki}\neq0$}
    \State {\label{line:dbw}}%
    $w_{ki} \gets a_{ki} \, \Deg_i^A \, \Deg_k^B$
    %$w_{ki} \gets \Deg_i^A \, \Deg_k^B$
    \EndFor
    \State $X \gets$ all-zero matrix of size $m \times n$
    \For{$\ell=1,\dots,s$}
    \State {\label{line:dbs1}}%
    Sample $(k,i)$ with probability $w_{ki} /\|W\|_1$
    \State {\label{line:dbs2}}% 
    Sample $j$ from $\mathcal{N}^B_k$
    \State {\label{line:dbs3}}% 
    Sample $k'$ from $\mathcal{N}^A_i$
    \State {\label{line:dbs4}}% 
    $x_{ij} \gets x_{ij} +  b_{k'j}$
    \EndFor
    \State Postprocessing (see \cref{alg:post})
  \end{algorithmic}
  \caption{Diamond sampling with binary inputs}
  \label{alg:diamond-binary}
\end{algorithm}

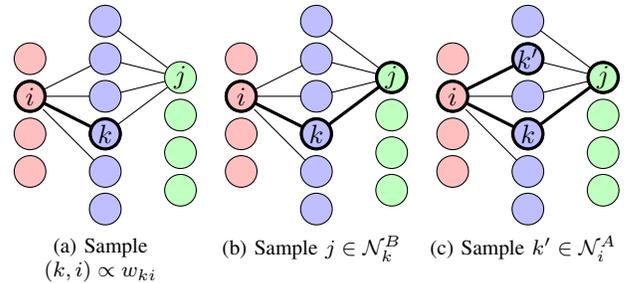
\begin{figure}[htbp]
  \centering
  \subfloat[Sample $(k,i) \propto w_{ki}$]{\label{fig:step1}
  \begin{tikzpicture}
    \TpgEdges
    \node (k) at (1,0.5*3-0.5) {};
    \draw [very thick] (k)--(i);
    \TpgNodes
    \node at (k) {$k$};
    \node [shape=circle,draw=black,minimum size=4mm,very thick] at (i) {};
    \node [shape=circle,draw=black,minimum size=4mm,very thick] at (k) {};
  \end{tikzpicture}}
~
  \subfloat[Sample $j \in \mathcal{N}^B_k$]{\label{fig:step2}
  \begin{tikzpicture}
    \TpgEdges
    \node (k) at (1,0.5*3-0.5) {};
    \draw [very thick] (k)--(i);
    \draw [very thick] (k)--(j);
    \TpgNodes
    \node at (k) {$k$};
    \node [shape=circle,draw=black,minimum size=4mm,very thick] at (i) {};
    \node [shape=circle,draw=black,minimum size=4mm,very thick] at (j) {};
    \node [shape=circle,draw=black,minimum size=4mm,very thick] at (k) {};
  \end{tikzpicture}}
~
  \subfloat[Sample $k' \in \mathcal{N}^A_i$]{\label{fig:step3}
  \begin{tikzpicture}
    \TpgEdges
    \node (k) at (1,0.5*3-0.5) {};
    \node (kp) at (1,0.5*5-0.5) {};
    \draw [very thick] (k)--(i);
    \draw [very thick] (k)--(j);
    \draw [very thick] (kp)--(i);
    \TpgNodes
    \node at (k) {$k$};
    \node at (kp) {$k'$};
    \node [shape=circle,draw=black,minimum size=4mm,very thick] at (i) {};
    \node [shape=circle,draw=black,minimum size=4mm,very thick] at (j) {};
    \node [shape=circle,draw=black,minimum size=4mm,very thick] at (k) {};
    \node [shape=circle,draw=black,minimum size=4mm,very thick] at (kp) {};
  \end{tikzpicture}}
\caption{\small Illustration of diamond sampling in \cref{alg:diamond-binary}. For simplicity, we show
  only those edges incident nodes $i$ or $j$.}
  \label{fig:diamond}
\end{figure}

The largest values in $X$ correspond to the (likely) largest dot
products, but we do some further postprocessing to obtain the final
answer, as shown in \cref{alg:post}.
We are seeking the top-$t$ dot products. We have a budget of $t'\geq
t$ dot products, where we assume $t' \ll mn$.
We let $\Omega_s$ denote the indices of all the nonzeros in $X$ and
$\Omega_{t'}$ denote the top-$t'$ entries in $X$; this requires a sort
in \cref{line:post1} of at most $s$ items (and generally many fewer,
depending on the proportion of three-paths that close into diamonds). We compute the $t'$ dot products in
\cref{line:post2a,line:post2b,line:post2c} at a cost of $O(t'd)$. 
Finally, we let $\Omega_t$ denote the top-$t$ dot products from
$\Omega_{t'}$ in \cref{line:post3}, requiring a sort of $t'$ items.

\begin{algorithm}
  Given $\Omega_s = \set{ (i,j) | x_{ij} > 0}$.
  Let $t$ be the number of top dot products, and $t' \geq t$ be the
  budget of dot products.
  \begin{algorithmic}[1]
    \State {\label{line:post1}}% 
    Extract top-$t'$ entries of $X$, i.e., 
    $|\Omega_{t'}| \leq t'$ and
    \Statex 
    $\Omega_{t'} \gets \set{ (i,j) \in \Omega_s | x_{ij} \geq
      x_{i'j'} \forall (i',j') \in \Omega_s \setminus \Omega_{t'}}$
    \State $C \gets$ all-zero matrix of size $m \times n$
    \For{$(i,j) \in \Omega_{t'}$} {\label{line:post2a}}% 
    \State {\label{line:post2b}}% 
    $c_{ij} \gets a_i^Tb_j$
    \EndFor {\label{line:post2c}}% 
    \State {\label{line:post3}}% 
    Extract top-$t$ entries of $C$, i.e., 
    $|\Omega_{t}| \leq t$ and
    \Statex 
    $\Omega_{t} \gets \set{ (i,j) \in \Omega_{t'} | c_{ij} \geq
      c_{i'j'} \forall (i',j') \in \Omega_{t'} \setminus \Omega_{t}}$
  \end{algorithmic}
  \caption{Postprocessing}
  \label{alg:post}
\end{algorithm}

\subsection{General inputs}

We present the binary version as general motivation, but our
implementation and analysis are based on the 
diamond sampling algorithm for general real-valued $A$ and $B$ 
in \cref{alg:diamond}. In this case, we define the matrix of
weights $W \in \Real^{d \times n}$ such that
\begin{displaymath}
  w_{ki} = |a_{ki}|\, \|\MC{A}{i}\|_1  \|\MR{B}{k}\|_1
  \qtext{for all} k \in [d], \, i \in [m].
\end{displaymath}
The weight $w_{ki}$ correspond to the weight of all three paths with edge
$(i,k)$ at its center. This is computed in
\cref{line:dw}.
The sampling in \cref{line:ds1} has the same complexity as in the
binary case, but the sampling in \cref{line:ds2,line:ds3} now has a
nonuniform distribution and so has higher complexity than in the
binary case.
The postprocessing is unchanged.
\begin{algorithm}
  Given matrices $A \in \Real^{m \times d}$ and $B \in
  \Real^{n \times d}$. \\
  Let $s$ be the number of samples.
  \begin{algorithmic}[1]
    %\For{$(k,i) \in [d] \otimes [m]$}
    \ForAll{$a_{ki}\neq0$}
    \State {\label{line:dw}}%
    $w_{ki} \gets |a_{ki}|\, \| \MC{A}{i} \|_1 \| \MR{B}{k} \|_1$
    \EndFor
    \State $X \gets$ all-zero matrix of size $m \times n$
    \For{$\ell=1,\dots,s$}
    \State {\label{line:ds1}}%
    Sample $(k,i)$ with probability $w_{ki} /\|W\|_1$
    \State {\label{line:ds2}}%
    Sample $j$ with probability $|b_{kj}| / \| \MR{B}{k} \|_1$
    \State {\label{line:ds3}}%
    Sample $k'$ with probability $|a_{k'i}| / \| \MC{A}{i} \|_1$
    \State {\label{line:dx}}%
    $x_{ij} \gets x_{ij} + \sgn(a_{ki} b_{kj} a_{k'i}) \, b_{k'j}$
    \EndFor
    \State {\label{line:post}}%
    Postprocessing (see \cref{alg:post})
  \end{algorithmic}
  \caption{Diamond sampling with general inputs}
  \label{alg:diamond}
\end{algorithm}

\subsubsection{Nonnegative inputs}

If $A$ and $B$ are nonnegative, the only change is that the sign
computations can be ignored in computing the sample increment in
\cref{line:dx} in \cref{alg:diamond}.
This avoids potentially expensive random memory accesses.

\subsubsection{Equal inputs (Gram matrix)}
\label{sec:gram}

If $B=A$, then $C=A^TA$ is symmetric. The matrix $X$ is not symmetric,
although $\Exp{X}$ is. Hence, we modify $X$ before by inserting the
following step before the postprocessing in \cref{line:post} in \cref{alg:diamond}:
\begin{equation}\label{eq:xsym}
  X \gets (X + X^T) / 2.
\end{equation}
Now $X$ is symmetric, and the forthcoming analysis is unaffected.

\subsubsection{Equal symmetric inputs (squared matrix)}
\label{sec:sym}

If $B=A$ and $A$ is symmetric, then $C=A^2$ and we can replace \cref{line:dx} in
\cref{alg:diamond} with the following two lines:
\begin{align*}
  x_{ij} & \gets x_{ij} + \sgn(a_{ki} b_{kj} a_{k'i}) \, b_{k'j} / 2,\\
  x_{kk'} & \gets x_{kk'} + \sgn(a_{ki} b_{kj} a_{k'i}) \, b_{k'j} / 2.
\end{align*}
This exploits the fact that we can swap the role of $k$ and $i$ in the
initial edge sample.
Again, $X$ may not be symmetric, so we 
insert \cref{eq:xsym} before the postprocessing in \cref{line:post}.

\subsection{Complexity and space}

Let $\alpha = \nnz{A}$ and $\beta = \nnz{B}$. In the dense case,
$\alpha = md$ and $\beta = nd$.
The total work is 
\begin{displaymath}
O(\alpha+\beta+s\log(s\alpha\beta)).
\end{displaymath}
The total storage (not counting the inputs $A$ and $B$) is
\begin{displaymath}
  2\,\storage{A} + \storage{B} + 5s + 3t' + 3t.
\end{displaymath}
We give detailed arguments below and in the implementation discussion in \cref{sec:implementation}.

\textbf{Preprocessing.}
For the sampling in \cref{line:dbs2,line:dbs3},
we precompute cumulative, normalized column sums for $B$ and the
same for rows of $A$, requiring storage of
$\storage{A}+\storage{B}$ and computation of $O(\alpha+\beta)$.
The matrix $W$ has the same nonzero pattern as $A$, so 
the cost to store it
is equal to $\storage{A}$ and to
compute it is $O(\alpha)$.

\textbf{Sampling.}
For a straightforward implementation, the cost per sample in
\cref{line:ds1} is $O(\log(\alpha))$.
For \cref{line:ds2}, the cost per sample is
$O(\log(\beta/d))$;
here, we have used the approximation 
$\nnz{\MR{B}{k}} \approx \beta/d$.
A similar analysis applied for $A$ and
\cref{line:ds3}.
So, the cost per sample is 
$O(\log(\alpha) +\log(\beta/d)+ \log(\alpha/m))$.
Without loss of generality, we assume that we need to store the
three-paths and the summand in \cref{line:dbs4} for
a total storage of $5s$. 

\textbf{Postprocessing.}
Conservatively, we require $3t'$ storage for the $(i,j,x_{ij}
\text{ or } c_{ij})$
triples in $\Omega_{t'}$ and $3t$ storage for the $(i,j,c_{ij})$
triples in $\Omega_{t}$. 
The sorting requires at most $O(s\log s)$ time, and usually much less
since $\nnz{X}$ may be much less than $s$ due to only some three-paths
forming diamonds and concentration, i.e., picking the same $(i,j)$
pair multiple times.

\section{Analysis of Diamond Sampling}
\label{sec:analysis}

This section provides a theoretical analysis of diamond sampling.
 We first prove that  the expected value of $x_{ij}$ is $c_{ij}^2 / \|W\|_1$, and then we prove  error bounds on our estimate as a function of the number of samples. 
Unless stated otherwise, our analysis applies to the general version
of the diamond-sampling algorithm (\cref{alg:diamond}).

\subsection{Expectation}

For a single instance of \cref{line:ds1,line:ds2,line:ds3} of
\cref{alg:diamond}, we define the event
\begin{displaymath}
 \mathcal{E}_{k'ikj} = \text{choosing three-path $(k',i,k,j)$}.
\end{displaymath}

\begin{lemma} 
  $\Prob(\mathcal{E}_{k'ikj}) =  | a_{ki} b_{kj} a_{k'i}| / \|W\|_1$.
\end{lemma}
\begin{proof} The probability of choosing three-path $(k',i,k,j)$ is (by independence
of these choices) the product of the following probabilities:
that of choosing the center edge $(i,k)$, then picking $j$, and then picking $k'$.
  \begin{align*}
    \Prob(\mathcal{E}_{k'ijk}) 
    &= \Prob(\text{ctr $(i,k)$}) \cdot
    \Prob(\text{endpts $j$ and $k'$} | \text{ctr $(i,k)$})  \\
    &= \frac{w_{ki}}{\|W\|_1} \cdot 
    \frac {|b_{kj}|}{\|\MR{b}{k}\|_1} \cdot 
    \frac {|a_{k'i}|}{\|\MC{a}{i}\|_1}\\
    &= \frac{ |a_{ki}|\, \| \MC{A}{i} \|_1 \| \MR{B}{k} \|_1}{\|W\|_1} \cdot 
    \frac {|b_{kj}|}{\|\MR{b}{k}\|_1} \cdot 
    \frac {|a_{k'i}|}{\|\MC{a}{i}\|_1} \\
    & = \frac{ |a_{ki} b_{kj} a_{k'i}| } { \|W\|_1}.
  \end{align*}
\end{proof}

In what follows, we use $X_{i,j,\ell}$ to be the following random variable:
if $i,j$ are the respective indices updated in the $\ell$th iteration,
$X_{i,j,\ell} =\sgn(a_{ki} b_{ki} a_{k'i})b_{k'j}$. Otherwise, $X_{i,j,\ell} = 0$.
Observe that $x_{ij} = \sum_{\ell=1}^s X_{i,j,\ell}$.

\begin{lemma}\label{lem:dexp}
 For diamond sampling,
  $\Exp{x_{ij} / s} = c_{ij}^2 / \|W\|_1$.
\end{lemma}
\begin{proof} We note that $\Exp{x_{ij}/s}
= \Exp{\sum_\ell X_{i,j,\ell}}/s = \Exp{X_{i,j,1}}$. (We use linearity of expectation and the 
fact that the $X_{i,j,\ell}$ are i.i.d. for fixed $i,j$ and varying $\ell$.)
  \begin{align*} 
    \Exp{X_{i,j,1}}
    & = \sum_{k} \sum_{k'} \Prob \bigl( \mathcal{E}_{k'ikj} \bigr) \cdot
    \sgn(a_{ki} b_{ki} a_{k'j} ) \, b_{k'j} \\
    & = \sum_{k} \sum_{k'} \frac{ |a_{ki} b_{kj} a_{k'i}| } { \|W\|_1} \cdot
    \sgn(a_{ki} b_{kj} a_{k'i}) \, b_{k'j} \\
    & = \frac{1}{\|W\|_1} 
    \sum_k \sum_{k'} a_{ki} b_{kj} a_{k'i}  b_{k'j} \\
    & = \frac{1}{\|W\|_1} 
    \Bigl( \sum_k a_{ki} b_{kj} \Bigr) 
    \Bigl( \sum_{k'} a_{k'i} b_{k'j}  \Bigr) \\
    &= \frac{1}{\|W\|_1} 
    \Bigl( \sum_k a_{ki} b_{kj} \Bigr) ^2 = \frac{c_{ij}^2}{\|W\|_1}.
  \end{align*}
\end{proof}

\subsection{Concentration bounds}

We now provide some concentration bounds
when all entries in $A$ and $B$ are nonnegative.

\begin{lemma} \label{lem:conc} Fix $\eps > 0$ and error probability $\delta \in (0,1)$.
Assume all entries in $A$ and $B$ are nonnegative and at most $\bound$.
If the number of samples 
\[s \geq 3\bound \|W\|_1 \log(2/\delta)/(\eps^2 c^2_{ij}),\]
then 
\[ \Pr[| x_{ij}\|W\|_1/s - c^2_{ij}| > \eps c^2_{ij}|] \leq \delta.\]
\end{lemma}

\begin{proof} Observe that $X_{i,j,\ell}$ is in the range $[0,\bound]$. Thus,
$Y_{i,j,\ell} = X_{i,j,\ell}/\bound$ is in $[0,1]$. Set $y_{ij} = \sum_{\ell} Y_{i,j,\ell}$.
Since $y_{ij}$ is the sum of random variables in $[0,1]$,
we can apply the standard multiplicative Chernoff bound (Theorem 1.1 of~\cite{DuPa}).
This yields $\Pr[y_{ij} \geq (1+\eps)\Exp{y_{ij}}] < \exp(-\eps^2\Exp{y_{ij}}/3)$.
By \cref{lem:dexp}, $\Exp{y_{ij}} = (s/\bound)(c^2_{ij}/\|W\|_1)$, which
is at least $3\log(2/\delta)/\eps^2$ by choice of $s$.
Hence, $\Pr[y_{ij} \geq (1+\eps)\Exp{y_{ij}}] < \delta/2$. 
Note that $y_{ij} = x_{ij}/\bound$.
We multiply the expression inside the $\Pr[\cdot]$ by $\bound\|W\|_1/s$
to get the event $x_{ij}\|W\|_1/s \geq (1+\eps)c^2_{ij}$. 

Using the Chernoff lower tail bound and identical reasoning, we get $\Pr[x_{ij}\|W\|_1/s \leq (1-\eps)c^2_{ij}] \leq \delta/2$.
A union bound completes the proof.
\end{proof}

The following theorem gives a bound on the number of samples required to distinguish ``large" dot products
from ``small" ones. The constant $4$ that appears is mostly out of convenience; it can be replaced
with anything $>\! 1$ with appropriate modifications to $s$.

\begin{theorem} \label{thm:sep} Fix some threshold $\tau$ and error probability $\delta \in (0,1)$. Assume all entries
in $A$ and $B$ are nonnegative and at most $\bound$. Suppose $s \geq 12\bound \|W\|_1 \log(2mn/\delta)/\tau^2$.
Then with probability at least $1-\delta$, the following holds for all indices $i,j$ and $i',j'$: 
if $c_{ij} > \tau$ and $c_{i'j'} < \tau/4$, then $x_{ij} > x_{i'j'}$.
\end{theorem}

\begin{proof} First consider some dot product $c_{ij}$ with value at least $\tau$. 
We can apply \cref{lem:conc} with $\eps = 1/2$ and error probability $\delta/mn$,
so with probability at least $1-\delta/mn$, $x_{ij}\|W\|_1/s \geq c^2_{ij}/2 \geq \tau^2/2$.
Now consider dot product $c_{i'j'} < \tau/3$. Define $y_{i'j'}$ and $Y_{i',j',\ell}$
as in the proof of \cref{lem:conc}. We can apply the lower tail bound of Theorem 1.1 (third part) of~\cite{DuPa}:
for any $b > 2e\Exp{y_{i'j'}}$, $\Pr[y_{i'j'} > b] < 2^{-b}$. 

We set $b \!= \!s\tau^2/2\bound \|W\|_1$. From \cref{lem:dexp} and the assumption that $c_{i'j'} \!< \!\tau/3$ and $\Exp{y_{i'j'}} \!= \!\Exp{x_{i'j'}}/\bound
\!= \!sc^2_{i'j'}/\bound \|W\|_1 \!\leq \!s\tau^2/(16\bound \|W\|_1) \!< \!b/2e$.
Plugging in our bound for $s$, $b \! \geq \! (12\bound\|W\|_1 \log(2mn/\delta)/\tau^2)\cdot \tau^2/(2\bound \|W\|_1)$
$= \!6\log(2mn/\delta)$. Hence, $\Pr[y_{i'j'}\! > \!b]\! < \!\delta/(2mn)$. 
Equivalently, $\Pr[x_{i'j'}\|W\|_1/s \! >\!  \tau^2/2] \!< \!\delta/(2mn)$.
We take a union bound over all the error probabilities (there are at most $mn$
pairs $i,j$ or $i',j'$). 

In conclusion, with probability at least $1-\delta$, for any pair of indices $i,j$:
if $c_{ij} > \tau$, then $x_{ij}\|W\|_1/s \geq \tau^2/2$. If $c_{ij} < \tau/4$,
then $x_{ij}\|W\|_1/s < \tau^2/2$. This completes the proof. 
\end{proof}

To get a useful interpretation of \cref{lem:conc} and \cref{thm:sep},
we ignore the parameters $\eps$ and $\delta$. Let us also assume that $K=1$,
which is a reasonable assumption for most of our experiments. Basically,
to get a reasonable estimate of $c_{ij}$, we require $\|W\|_1/c^2_{ij}$ samples.
If the value of the $t$-th largest entry in $C$ is $\tau$, we require
$\|W\|_1/\tau^2$ samples to find the $t$-largest entries.  For instance, on a graph, if we want to identify  pairs of vertices with at least 200 common neighbors, we can set $\tau=200$, and $\|W\|_1$ will be the number of (non-induced) 3-paths in the graph. 
The square in the denominator is what makes this approach work. 
In \cref{tab:stats} of \cref{sec:experiments}, we show some of the values of $\|W\|_1/\tau^2$ for particular datasets, where $\tau$ is the magnitude of the largest entry.

\section{Implementation Details}
\label{sec:implementation}

We discuss the implementation details for reproducibility, but we stress that the implementation is not our primary contribution.
Nevertheless, careful thought has gone into the process and 
we show that a clever implementation of the sampling can improve performance by almost $3\times$; see \cref{fig:perf_LO}.

\subsection{Sampling from discrete distributions}
\label{sec:sampl_disc}

We consider two alternative schemes for drawing $s$ samples from an
arbitrary discrete distribution defined by the vector $\vec{\rho} \in
[0,1]^p$ such that $\sum_{k=1}^p \rho_k = 1$.
The choice of schemes is based on the relative sizes of $s$ and $p$.

If the size of the distribution, $p$, is smaller than the number of
samples, $s$, then we use binary search on the cumulative sums of 
$\rho$ to determine each sample. This requires $O(s\log p)$ comparisons, plus
$O(p)$ work for the preprocessing to compute the cumulative sum.
We note that using the alias method \cite{Vo91} yields a constant time per search at the same cost for preprocessing and storage (up to a constant).
The cumulative sum requires $O(p)$ space, and sampled events are stored as counts in space $O(p)$.
\hideme{
\Note{TGK:We can remove the algorithm and the next sentence if space
  is a problem.}
\cref{alg:sampling_smalln} shows the desired approach for $p<s$.
}
Note that both the binary search and sample counter increments involve random (not contiguous) memory access.
This returns a count vector $\vec{c}$ of length $p$ such that $c_k$ is the number of
occurrences of event $k$ and $\sum c_k = s$.

\hideme{
\begin{algorithm}
  Given probability distribution $\vec{\rho} \in [0,1]^p$ and \\
  $s$ = number of samples
  \begin{algorithmic}[1]
    \State $\bar\rho_0 \gets 0$
    \For{$k=1,\dots,p$}
    \State $\bar\rho_k \gets \bar\rho_{k-1}+\rho_k$
    \Comment cumulative sum
    \EndFor
    \State $\vec{c} \gets 0$ \Comment{sample count $\vec{c} \in \Real^p$}
    \For{$\ell=1,\dots,s$}
      \State $r \gets U(0,1)$
      \State Find $k$ using binary search such that $\bar\rho_{k-1} < 
      r \leq \bar\rho_k$ 
      \State $c_k \gets c_k+1$
    \EndFor
  \end{algorithmic}
  \caption{Standard binary search (for $p<s$)}
  \label{alg:sampling_smalln}
\end{algorithm}
}
If, on the other hand, the number of samples, $s$, is less than the
size of the distribution, $p$, we can avoid the binary search by doing a
single sort of the samples, as shown in \cref{alg:sampling_smalls}.
This is essentially a variation on merge sort, but we sort only one of the two lists and compute the other on the fly.
The preprocessing involves sorting $s$ random numbers,
requiring $O(s\log s)$ comparisons and $O(s)$ space. 
Sampled events are determined by walking through the sorted samples
and the probability distribution, computing the cumulative sums along
the way, requiring $O(n+s)$ computations. 
Sampled events are stored explicitly in space $O(s)$.
Note that all memory accesses in this approach are contiguous (reads and writes).
This returns an explicit sample list $\vec{e}$ of length $s$ such that $e_1 \leq
\cdots \leq e_{\ell}$.

\begin{algorithm}\footnotesize
  Given probability distribution $\vec{\rho} \in [0,1]^p$ and 
  $s$ = \# samples
  \begin{algorithmic}[1]
    \State $r_{\ell} \gets U(0,1)$ for all $\ell \in [s]$
    \State Sort the vector $\vec{r}$ so that $r_1\leq\cdots\leq r_\ell$
    \State $k \gets 1$,
     $\bar \rho \gets \rho_k$ 
    \For{$\ell=1,\dots,s$}
      \While{$r_\ell>\bar\rho$}
        \State $k \gets k+1$,
         $\bar\rho \gets \bar\rho + \rho_k$
      \EndWhile
      \State $e_\ell \gets k$
    \EndFor
  \end{algorithmic}
  \caption{Search via sample sorting (for $s < p$)}
  \label{alg:sampling_smalls}
\end{algorithm}

Thus, the cost of sampling and storing $s$ events from a discrete distribution of size $p$ is $O(s \log (\min\{s,p\}) + p)$ computations and $O(\min\{s,p\})$ space.

\subsection{Diamond sampling with locality optimizations}
\label{sec:LO}

The implementation of \cref{alg:diamond} takes advantage of the
specialized sorting mentioned above as well as locality, as we
explain.

In the preprocessing, in anticipation of the sampling in
\cref{line:ds2}, we compute a matrix $\hat B$ such that each row is a
normalized cumulative sum, i.e., $\hat b_{kj} = \sum_{k'\leq k} b_{k'j}
/ \|\MR{b}{k}\|_1$.  We store the matrix $\hat B$ in compressed sparse
row (CSR) format, so that the entries of $\MR{\hat b}{k}$ are
contiguous in memory.
Similarily, for \cref{line:ds3}, we compute a matrix $\hat A$ such
that each column is a normalized cumulative sum, i.e., $\hat a_{ki} =
\sum_{i'\leq i} a_{ki'} / \|\MC{a}{i}\|_1$.  We store the matrix $\hat
A$ in compressed sparse column (CSC) format, so that the entries of
$\MR{\hat a}{i}$ are contiguous in memory.
Note that storage $\hat A$ in CSC format is equivalent to storage
$\hat A^T$ in CSR format, so we need only one data structure.

We separate lines \cref{line:ds1,line:ds2,line:ds3,line:dx} into four
separate loops, first computing $s$ pairs of the form $(k,i)$, then
$s$ $B$-neighbors, etc.

For the samples in \cref{line:ds1}, we use the search via sample
sorting in \cref{alg:sampling_smalls} for choosing the samples from
$W$ since typically $\nnz{W} = \nnz{A} \gg s$.
Because of the way that \cref{alg:sampling_smalls} works, the pairs
$(k_{\ell}, i_{\ell})$ for $\ell \in [s]$ are conveniently sorted according to $k_{\ell}$.

The sorted values yield data locality for \cref{line:ds2}, where we
use standard binary search to choose the values $j_{\ell}$ for $\ell
\in [s]$.

We rearrange the $s$ samples in $O(s)$ time so that they are ordered
according to $i_{\ell}$. Then we use standard binary search to choose
the values $k'_{\ell}$ for $\ell \in [s]$ from \cref{line:ds3}.

Finally, we reorder the samples in $O(s)$ time so that they are sorted
according to $k'_{\ell}$, enabling efficient lookups for $b_{k'j}$
values in \cref{line:dx}.

\hideme{
Locality optimizations for diamond sampling introduce some overheads;
compared with \cref{alg:diamond}, our implementation involves three extra reorderings.
The random number sorting involves an extra $O(s\log s)$ operations, and the two transposes each require $O(s)$ operations. 
Because these overheads depend only on the number of samples $s$ and not on the size of the input data, we expect the benefits from locality in accessing data arrays of size $\nnz{A}$ and $\nnz{B}$ will far outweigh the overheads.
For small $s$, the overheads should be negligible, and for large $s$, we expect to see more repeated values in the sampling and therefore greater benefit from locality.
}

\cref{fig:perf_LO} shows a 2.7 times speed-up for our optimized
implementation of \cref{alg:diamond} versus a straightforward implementation.
In particular, we note the drastic reduction in time for center, left, and right samples and setting the output entry (which involves searching for the existence of the 4th edge) in the optimized implementation.
This is due to achieving better data locality (i.e., cache performance), and the overheads of the reorderings to attain this locality are amortized.

\begin{figure}
\centering
\includegraphics[width=.8\columnwidth]{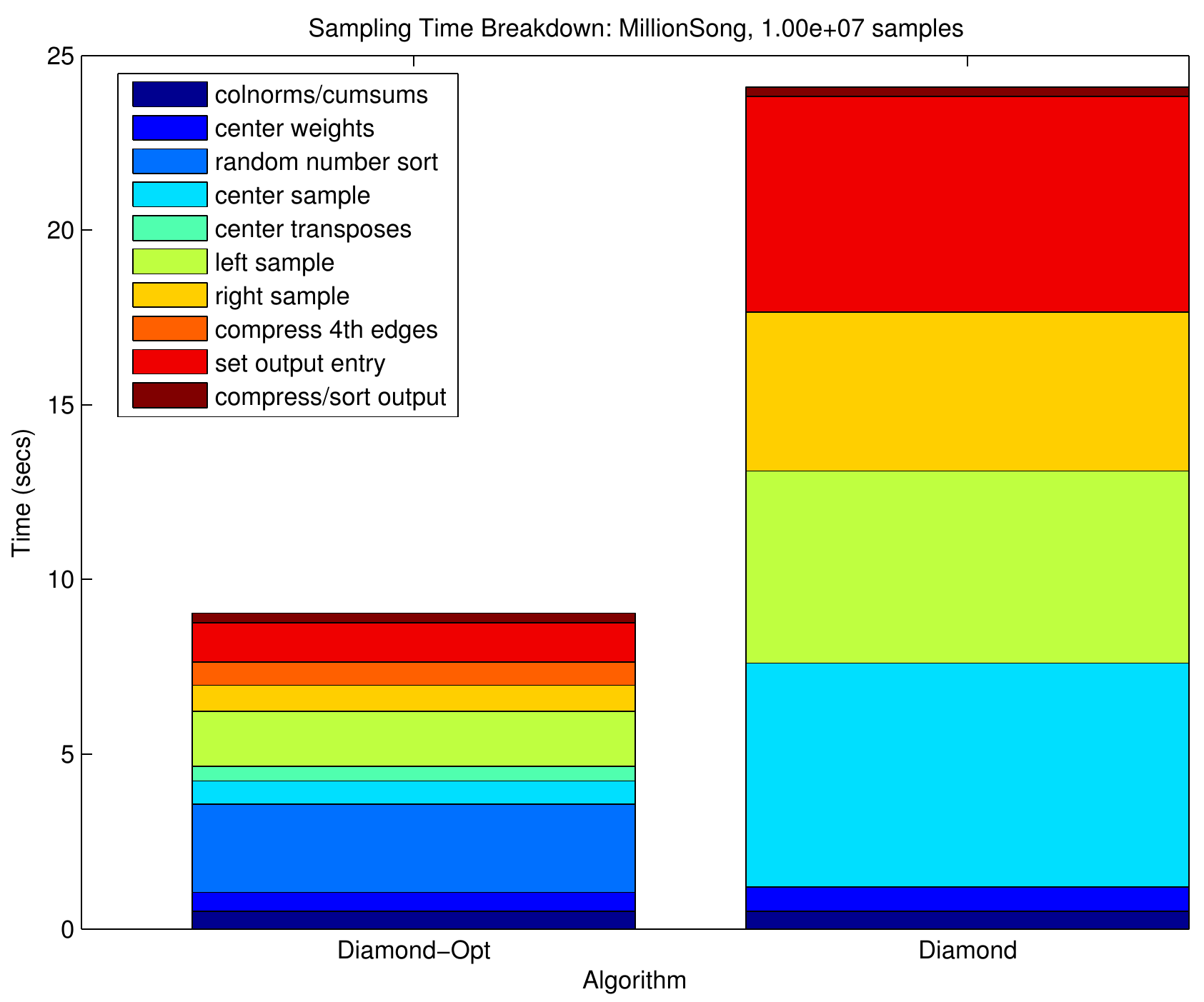}
\caption{\small Time breakdown for optimized and straightforward
  implementations of diamond sampling (\cref{alg:diamond}). 
}
\label{fig:perf_LO}
\end{figure}

\subsection{Exact computation (for comparison)}
\label{sec:exact-impl}

MAD corresponds to find the highest entries in a matrix-matrix
product. General high-performance implementations are available.
In the dense case, the BLAS interface allows access to vendor-tuned
libraries like Intel's Math Kernel Library \cite{MKL14} or NVIDIA's
cuBLAS \cite{CUBLAS}. 
In the sparse case, matrix multiplication is available in CSparse
\cite{Davis06}, an efficient open-source library that is used by
MATLAB for many sparse computations. 
The computational cost of matrix multiplication is $O(mnd)$ in the
dense case (assuming the classical algorithm is used) and $O(\sum_k \Deg_k^A \cdot \Deg_k^B)$, i.e., the number
of wedges in the tripartite graph, in the sparse case. 
The storage cost of library implementations of matrix multiplication,
$mn$ in the dense case and up to $mn$ in the sparse case, is generally
the limiting factor.

To adapt these high-performance libraries, we perform a series of
matrix-vector products, $\vec{c_j} = A^T\vec{b_j}$ for $j \in [n]$, to
compute the columns of $C$ one at a time. 
We do not save the columns but instead use a priority queue to track
the top-$t$ entries. 
Because CSparse is open source, we were able to modify the code to
minimize the memory footprint, achieving $O(\text{storage}(A)+\text{storage}(B)+t)$, with
little loss in performance. 
In the dense case, we compute $C$ in column blocks to size $n \times
d$, computing and processing the output matrix in chunks using the
\texttt{dgemm} interface to MKL, so that the memory footprint is
$O(\text{storage}(A)+\text{storage}(B)+t)$. 

\section{Experiments}
\label{sec:experiments}

All experiments are run on an Intel Xeon E5-2650 ``Ivy Bridge'' 2.0 GHz machine with 32 GB of memory.
Our codes are written in C/C++ with a mex-interface to MATLAB (Version 8.3.0.532). 
The codes are all single-threaded.

\subsection{Datasets}
\label{sec:datasets}

We experiment on  real-world datasets described below.

\begin{itemize}
	\item \emph{as-Skitter $A^TA$}: Skitter is an internet topology
          graph from the ``as-Skitter'' dataset from SNAP
          \cite{snapnets}.  This yields a sparse binary symmetric $n \times
          n$ matrix $A$ with  $n$ = 1,696,415 nodes and 11,095,298 nonzeros.
	\item \emph{Movielens $A^TB$}: The Movielens-10M data set \cite{movielens} comprises a sparse movie-user matrix, $R$, of size $m\times n$ with $m$ = 65,133 movies and $n$ = 71,567 users.  
Following Shrivastava and Li \cite{ShLi14}, who in turn followed \cite{CrKoTu10}, we compute the low-rank SVD of $R$ using $d=150$ components, so that $R \approx A^TB$ where $A \in \Real^{d \times m}$ and $B \in \Real^{d \times n}$ are dense real-valued matrices. 
	\item \emph{Live Journal $A^TA$}: LiveJournal is a free online
          community, and we use the ``soc-LiveJournal1'' dataset from
          SNAP \cite{snapnets}.  The corresponding ``friendship'' sparse symmetric binary
          adjacency matrix $A$ has dimension $n$ = 4,847,571  and 68,993,773 nonzeros.
	\item \emph{ASIC $A^TA$}: The ASIC dataset is a Xyce circuit simulation matrix; we use the ``Sandia/ASIC\_680k'' matrix from the Florida matrix collection \cite{DaHu11}.  The $n \times n$ real-valued matrix $A$ is nonsymmetric (though it is structurally symmetric) with dimension $n$ = 682,862 and 2,638,997 nonzeros.
	\item \emph{Amazon Kindle $A^TB$}: The Amazon product data
          consists of review and product data from Amazon
          \cite{McTaShHe15,amazon}. The Kindle category data includes
          $m$ = 1,406,916 reviewers and $n$ = 430,532 books. The
          reviewer-book rating matrix $A$ contains numeric scores from
          1--5 for reviewed books and has a total of 3,205,546 entries.
          The book-book similarity matrix $B$ contain numeric scores
          of 1--4 to indicate the relationship (i.e., 4 indicates
          the book have been purchased together by someone) with a
          total of 11,012,558 entries.
	\item \emph{Million Song $A^TA$}: The Echo Nest Taste Profile Subset of the Million Song Dataset contains 48M user-song play
counts from real users \cite{BeElWhLa11,McBeElLa12,echonest}. 
The resulting user-song matrix $A$ has 1,019,318 users, 384,546 songs, and 48,373,586 user-song play counts.
\end{itemize}

\subsection{Time and Accuracy Performance}

We present time and accuracy results for six data sets in \cref{fig:topids1,fig:topids2}. 
In this study, we vary $s$ as an axis, set $t'=s$, and plot results for $t\in\set{1,10,100,1000}$.
The top row plots the recall, i.e., the percentage of the top-$t$ entries identified, versus the number of samples, $s$. (Not all samples close to form diamonds; see \cref{tab:stats}.) 
The bottom row plots the wall-clock computation time versus the number of samples, including the time for exact computation as described in \cref{sec:exact-impl}.
Because the budget $t'$ does not depend on $t$ (we set $t'=s$), the timing is the same for all runs.
\Cref{tab:stats} contains some additional data about the sampling, including the size of the largest entry, the size of $\|W\|_1$, the ratio $\|W\|_1/\max c_{ij}^2$ (which is proportional to the number of samples needed to recover the largest entry according to \cref{thm:sep}), and the closure rate of the three paths to form diamonds.

\begin{table}[htpb]
\centering
\begin{tabular}{|c|cccc|} 
\hline 
Dataset & $\max |c_{ij}|$ & $\|W\|_1$ & est. samples & closure \\ \hline
as-Skitter & 3.0e5 & 2.9e12 & 3e3 & 19\% \\
Live Journal & 3.0e3 & 1.5e12 & 1.6e5 & 13\% \\
Movielens & 1.1e0 & 2.8e10 & 2.3e8 & 100\% \\
ASIC & 2e12 & 9.6e19 & 2.4e-5 & 100\% \\
Amazon Kindle & 3.2e3 & 1.2e12 & 1.2e5 & 1.4\% \\
Million Song & 6.1e6 & 1.4e15 & 3.8e1 & 15\% \\
\hline
\end{tabular}
\caption{\small Summary statistics for datasets.  The column labeled ``est. samples'' reports the ratio $\|W\|_1/\max c_{ij}^2$, the estimated number of samples required to find the top entry. The column labeled ``closure'' is the percentage of sampled three paths that correspond to a successful diamond sample.}
\label{tab:stats}
\end{table}

\begin{figure*}[htb]
\centering
\includegraphics[width=.3\textwidth]{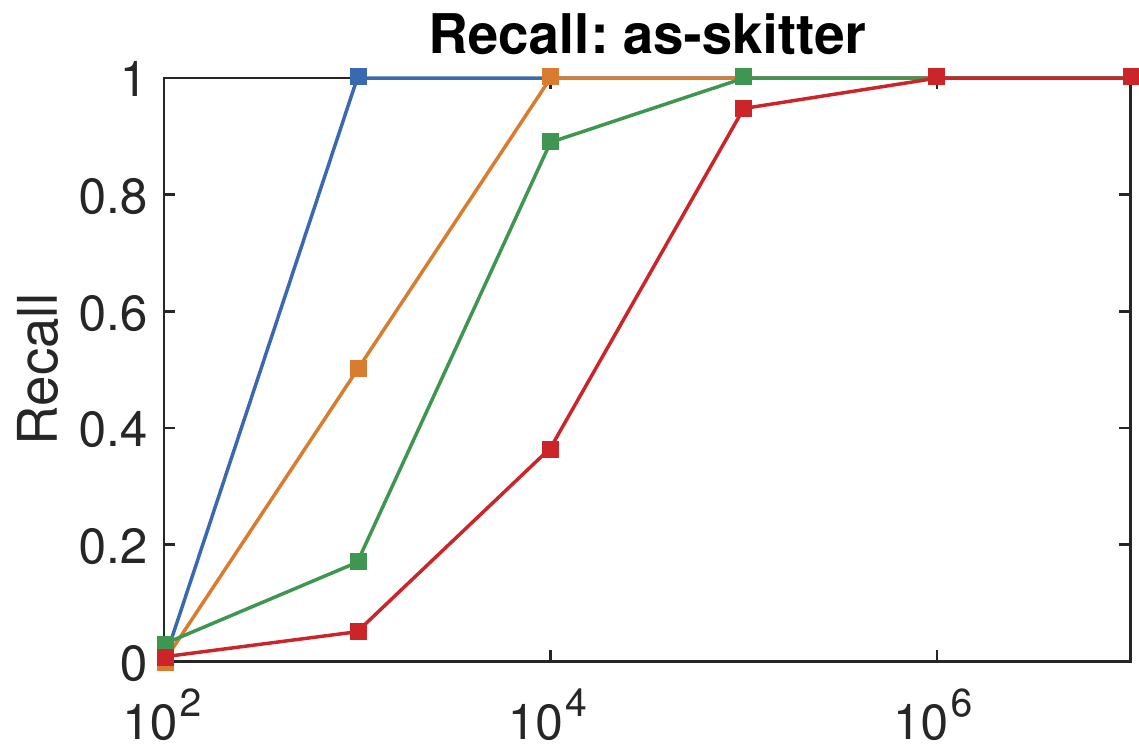} 
\includegraphics[width=.3\textwidth]{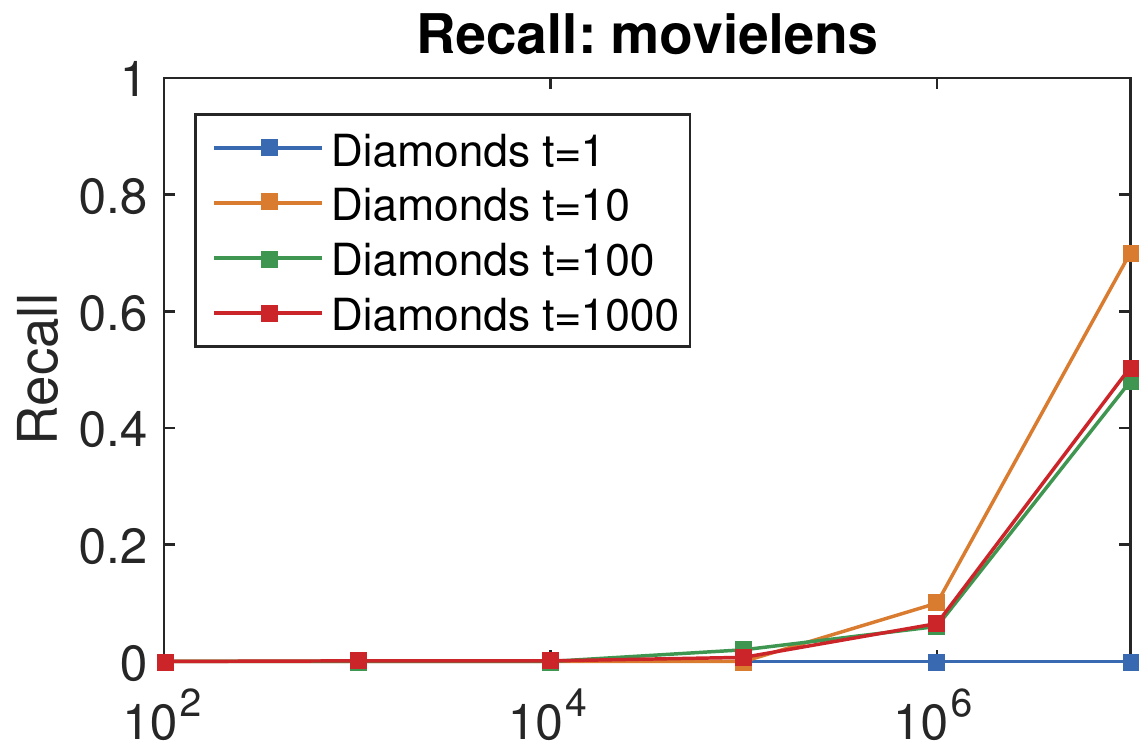} 
\includegraphics[width=.3\textwidth]{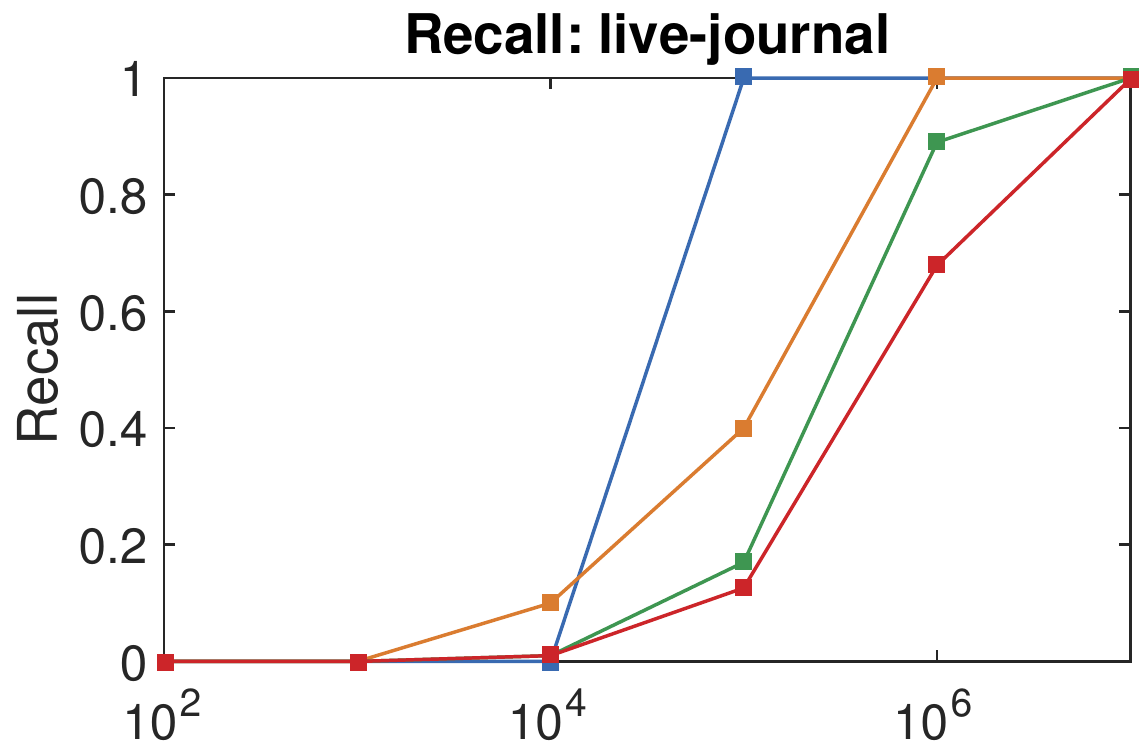} \\
\includegraphics[width=.3\textwidth]{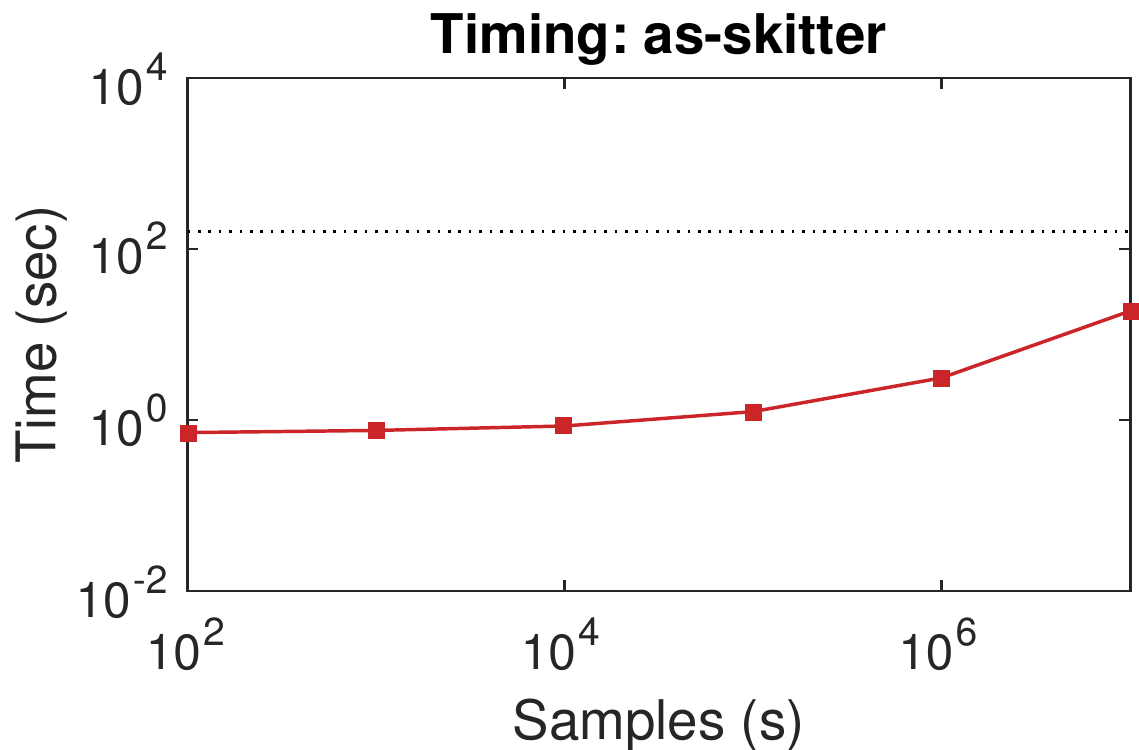}
\includegraphics[width=.3\textwidth]{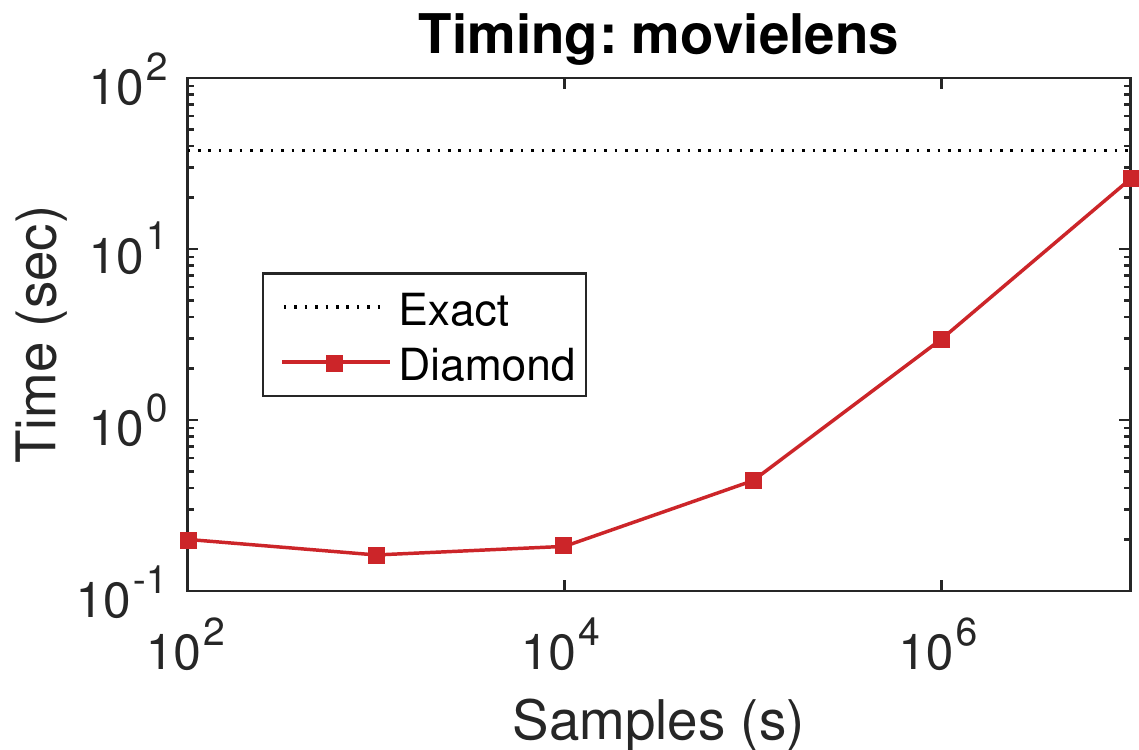} 
\includegraphics[width=.3\textwidth]{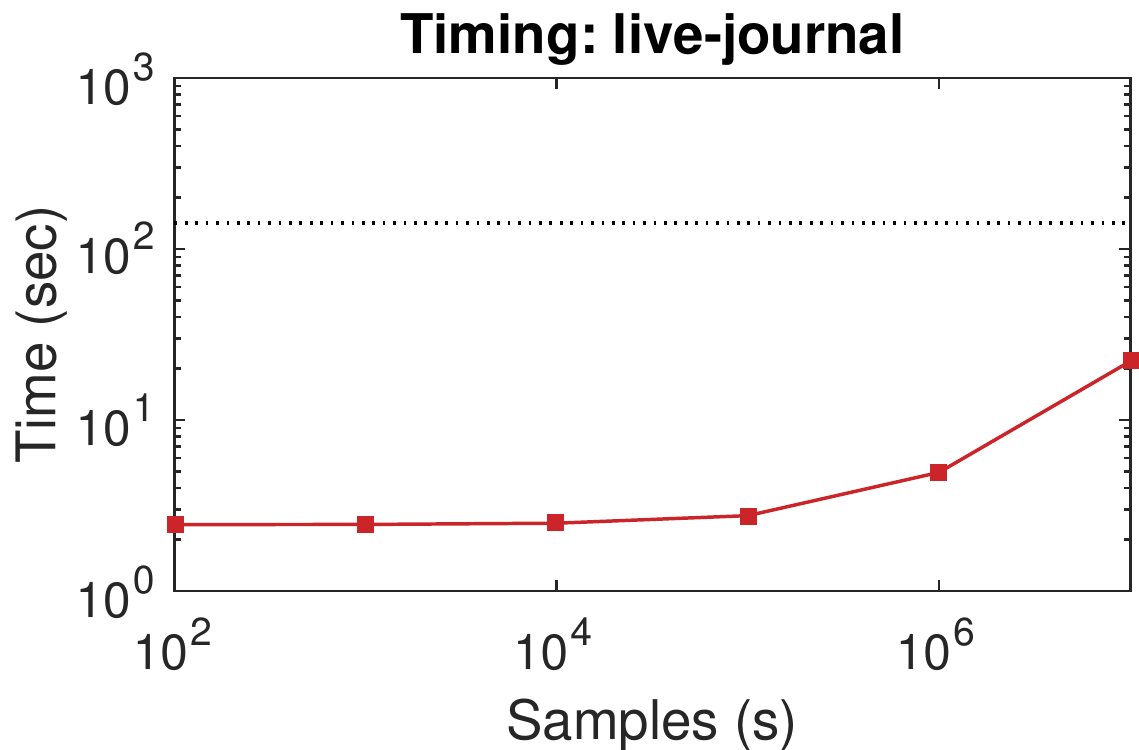} 
\caption{\small Time and accuracy results for sampling approaches for datasets as-Skitter, Movielens, and Live Journal.  The first row of plots presents the top-$t$ scores over various numbers of samples, and the bottom row of plots shows the time in logarithmic scale.}
\label{fig:topids1}
\end{figure*}

\begin{figure*}[htb]
\centering
\includegraphics[width=.3\textwidth]{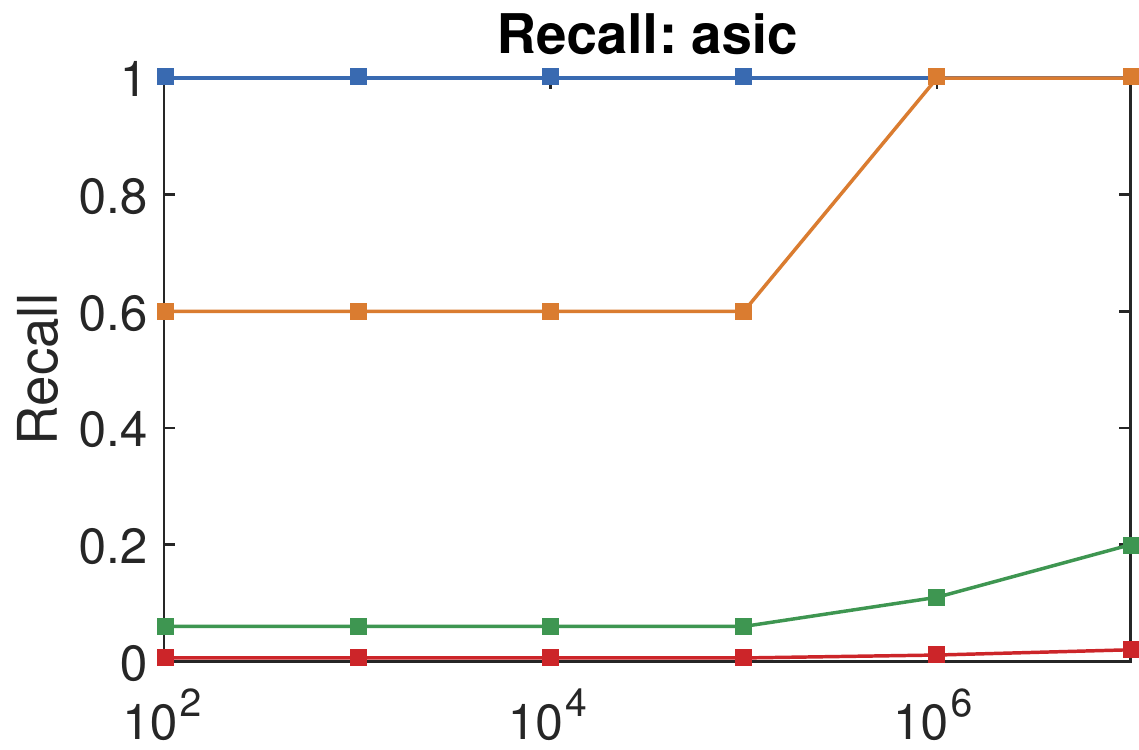} 
\includegraphics[width=.3\textwidth]{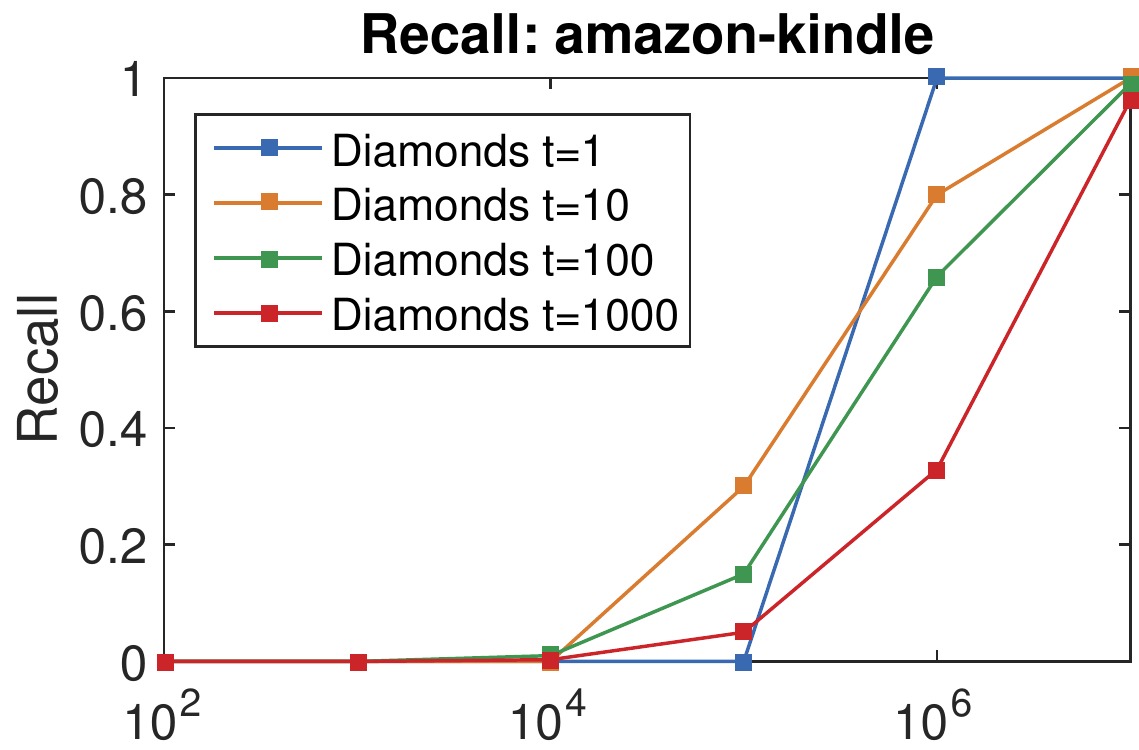} 
\includegraphics[width=.3\textwidth]{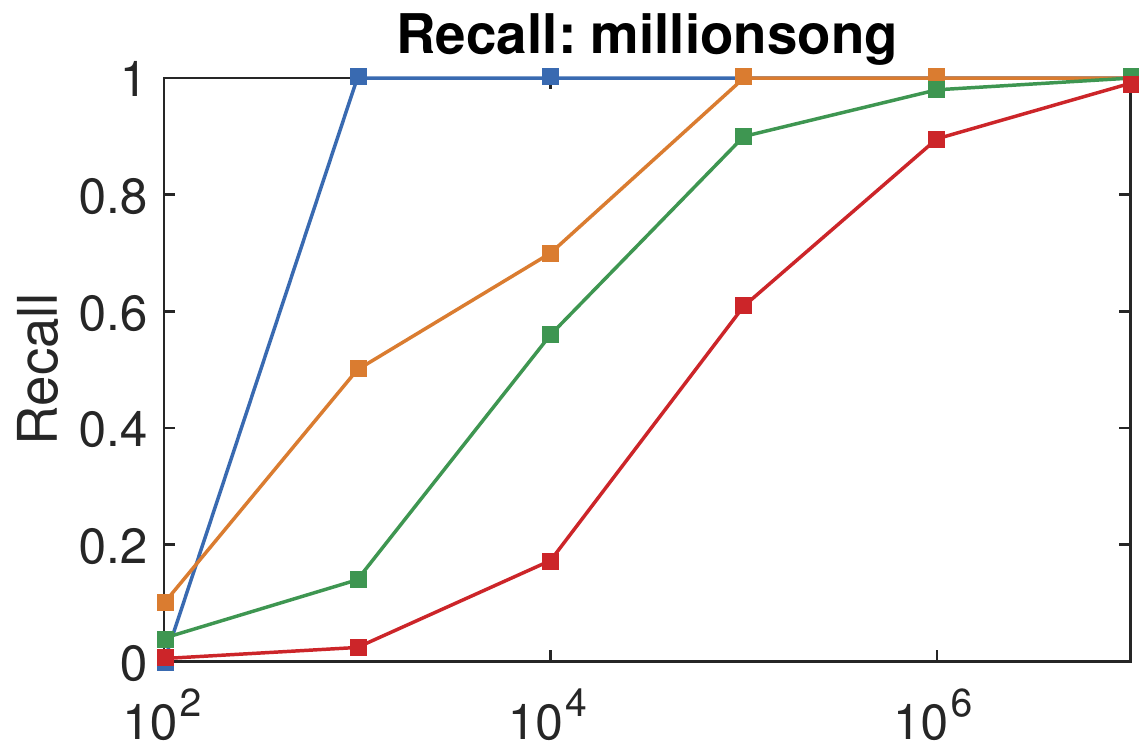} \\
\includegraphics[width=.3\textwidth]{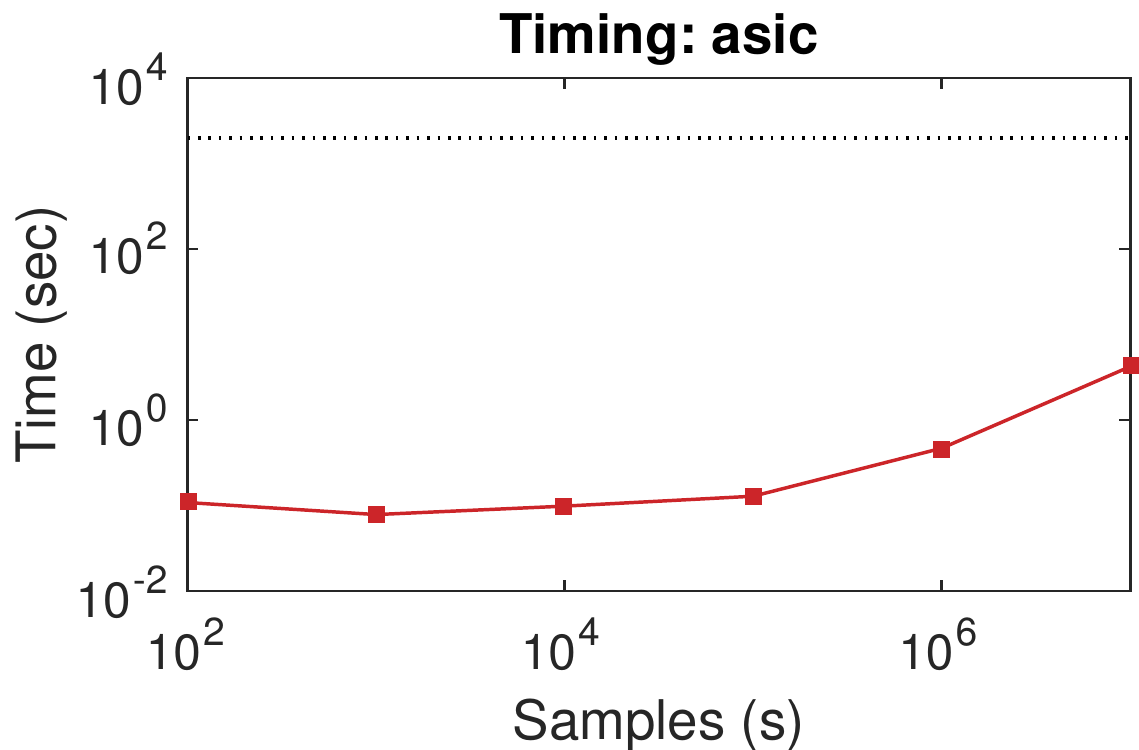}
\includegraphics[width=.3\textwidth]{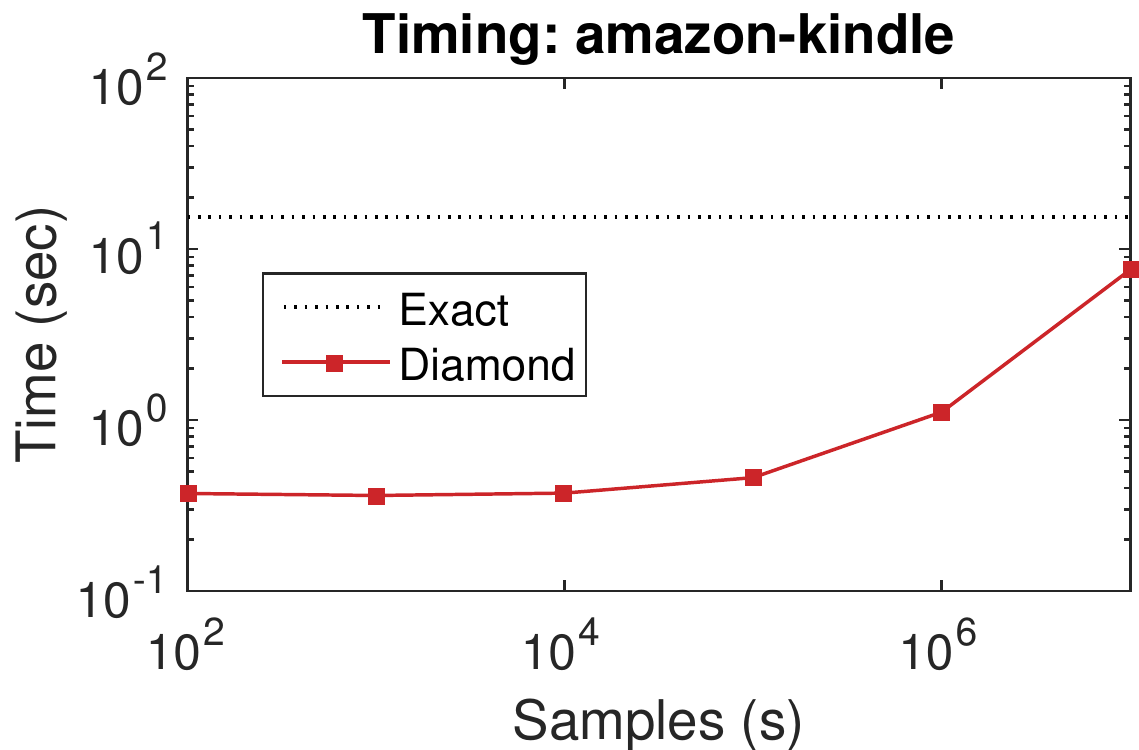} 
\includegraphics[width=.3\textwidth]{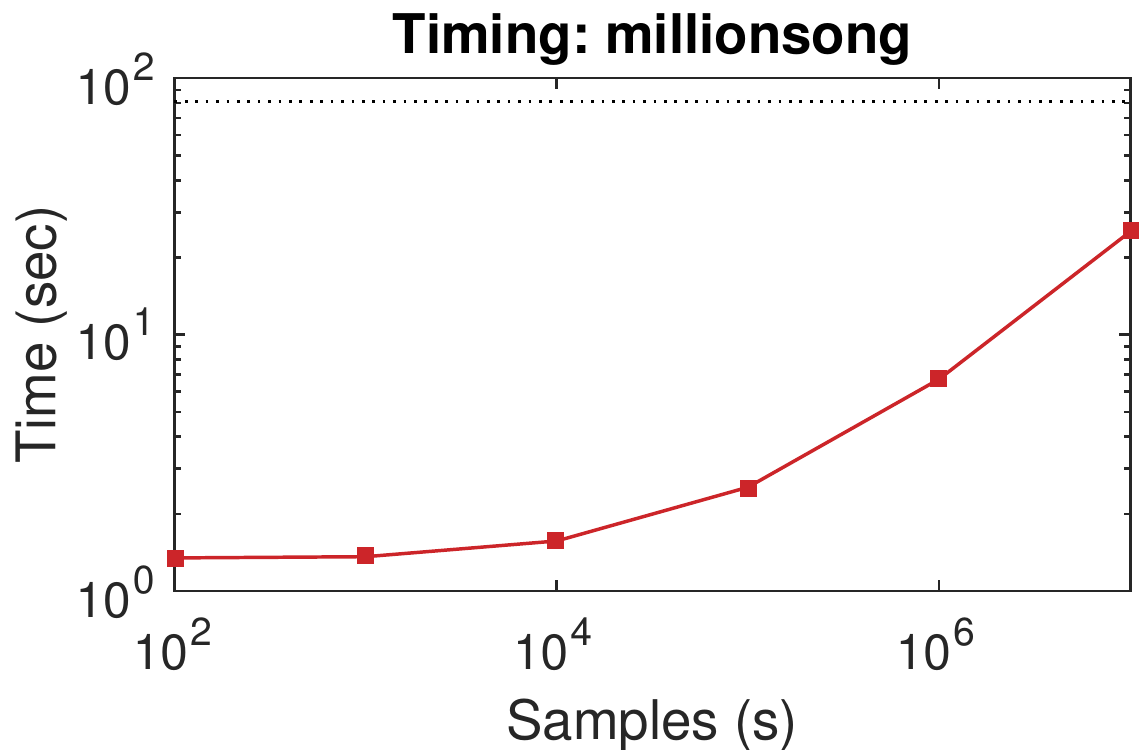} 
\caption{\small Time and accuracy results for sampling approaches for datasets ASIC, Amazon Kindle, and Million Song.  The first row of plots presents the recall scores over various numbers of samples, and the bottom row of plots shows the time in logarithmic scale.}
\label{fig:topids2}
\end{figure*}

For as-Skitter ($A^TA=A^2$ because input matrix $A$ is symmetric; see \cref{fig:topids1}), only $10^5$ samples are required to capture all top 1000 entries in the output matrix, which requires 1.25 seconds (dominated by the preprocessing step).
Exact computation, on the other hand, requires 160 seconds, which is 128 times slower.
The top entry of this matrix has value 30,620, while the number of three paths $\|W\|_1$ is 2.9e12.
The analysis in \cref{sec:analysis} suggests that approximately 3000 samples are required to find the top entry (see \cref{tab:stats}), and we identified the top entry after 1000 samples in this experiment.
Likewise, the 1000th top entry in $A^2$ has value 4,239, the analysis suggests 16K samples, and only 10K were needed for this experiment.
We find all top-1000 entries with only $10^6$ samples, and only 19\% of those samples turn into diamonds.

MovieLens (dense $A^TB$; see \cref{fig:topids1}) is the most difficult dataset because there is not much
differentiation between the largest entries and smaller ones: the
largest entry has magnitude 11.02 while the 1000th largest entry has
magnitude 7.37. Also, the estimated number of samples just to get the top entry is $10^8$; see \cref{tab:stats}.
Here, $10^7$ samples is not sufficient and anything more requires more time than exact computation.
The diamond closure rate is 100\% because the matrices are dense.
This is an example of a relatively small dataset where sampling is not effective; nevertheless, we
still have impressive precision-recall results in \cref{sec:comp-asymm-lsh}.

In LiveJournal (sparse $A^TA$; see \cref{fig:topids1}, we have $||W||_1=1.5 \times 10^{12}$ and the top entry is 2997, so we estimate needing $||W||/c_{ij}^2\approx10^5$ samples to find the largest entry (see \cref{tab:stats}), which is exactly when we find it. We find all top-1000 entries with $10^7$ samples and 10X less time than exact computation.

The ASIC graph (sparse $A^TA$; see \cref{fig:topids2}) comes from scientific computing. Here the largest few entries are very large compared to all others. For instance, the predicted number of samples is $< 1$ in \cref{tab:stats}.
So, we can identify the top-10 but struggle to identify the much smaller entries in the remainder of the top-1000 (10 million samples identifies only 703 distinct output entries). For the top-10, however, we have three orders of magnitude speed-up compared to exact computation.

We can find nearly the top-1000 for Amazon Kindle (sparse $A^TB$, see \cref{fig:topids2}) using $10^7$ samples, but the time is coming somewhat close to exact computation. This is a relatively small problem, and we expect that larger problems will have a more significant benefit. The performance on the recommendation application in \cref{fig:amazon-kindle-rec} yields good results.

We have 10X speed-up for the top-1000 entries for Million Song (sparse $A^TA$, see \cref{fig:topids2}). We find the top entry after 1000 samples, which is a bit higher than the estimate in \cref{tab:stats} of $\approx 100$ samples.

\subsection{Applications of MAD}

\begin{scenario}[Free Samples]\label{free}
  A review site wants to encourage more reviews, so its goal is to
  select a limited number of reviewer-product pairs with the idea that
  the selected reviewer will be given a \emph{free} item to review.
  We let $A$ be a reviewer-by-product matrix and $B$ is a product
  ``also-bought'' matrix. We then pick reviewer-product pairs, with
  the caveat that we should not give a reviewer something that they
  have already reviewed.  
\end{scenario}

We use Amazon purchase data on reviews and product-product relationships \cite{McTaShHe15,amazon}.
We focus on the Amazon Kindle subset, in which case we recommend e-books to be given to particular reviewers to solicit their reviews.
Given the top recommendations for a free sample, we must check that the reviewer has not already reviewed the product; for this data set, about half of the top-$t$ pairs correspond to new products for the particular reviewer.
\cref{fig:amazon-kindle-rec} shows an example top entry: it recommends a romance book to a user that has already reviewed many similar romance novels (three examples selected at random are shown).

\begin{figure}[htb]
\centering
\begin{tabular}{|c|c|} \hline
Rec. Book & Other Books Reviewed by User \\
\includegraphics[width=.1\textwidth]{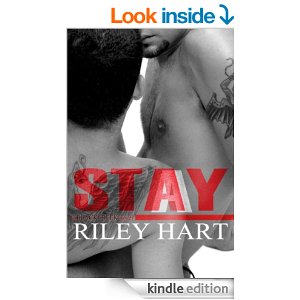}  &
\includegraphics[width=.1\textwidth]{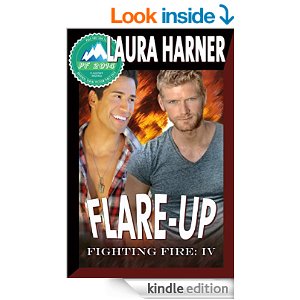} 
\includegraphics[width=.1\textwidth]{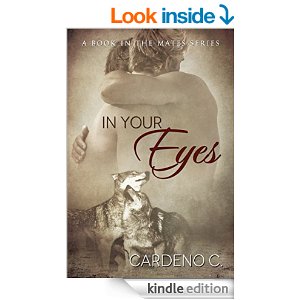} 
\includegraphics[width=.1\textwidth]{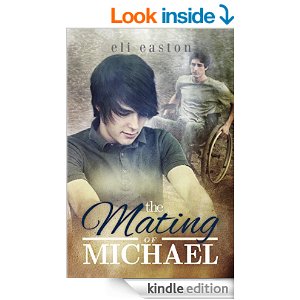} \\ \hline
\end{tabular}
\caption{\small Example top pair from Amazon Kindle dataset: user ID 1317513 and the book entitled ``Stay.''  The user has reviewed 649 other books, three random examples of which are shown.}
\label{fig:amazon-kindle-rec}
\end{figure}

\begin{scenario}[2For1]\label{2for1}
  A retailer wants to select, say, 100 pairs of products for a
  2-for-the-price-of-1 promotion.  If we assume each product has a
  dense or sparse representation in some feature space, this becomes a MAD search with $B=A$. In the case of the Million Song dataset, we let $A$ denote a
  song-by-user matrix where entry $(i,j)$ denotes the number of plays
  of song $i$ by user $j$. We
  want to find pairs of distinct songs that have the highest number of
  common plays. 
\end{scenario}

For this, we find the
top-$t$ pairs of similar songs in the Million Song dataset, based on being
played by the same users. 
We calculate the top-$t$ all-pairs dot
products for columns of $A$ as explained in
\cref{sec:gram}. Additionally, we ignore the diagonal entries that
pair songs with themselves.
The top song pair using this metric is ``Undo'' by Bj\"{o}rk and ``Revelry'' by Kings of Leon, which are both in the alternative genre.

\begin{scenario}[Link Prediction]
  We want to find members of a social network that should be connected but are not. These can be used as recommendations for new ``friends.''
\end{scenario}

We use the Live Journal data.  The top entry in $A^2$ is not an edge
in $A$; 6 out of top 10 are not; 55 out of top 100 are not; and 511
out of top 1000 are not. The top-10 new-connection suggestions have
over 1800 common neighbors per pair.

\subsection{Comparison to wedge sampling by  Cohen-Lewis}
\label{sec:wedge}

The most similar related work is that of Cohen and Lewis \cite{CoLe97,CoLe99}, which we refer to as wedge sampling.
In this section, we compare diamond sampling to wedge sampling and present results for the Skitter dataset, which showed the most distinction between the methods.
We generalized Cohen and Lewis' approach to $t$-MAD and implemented wedge sampling with similar optimizations to those described in \cref{sec:LO}.
In general, using diamonds requires fewer (three-path) samples than wedge samples to identify top entries in the output matrix.
In our implementation, the preprocessing cost for diamonds is greater than for wedges, but the per-sample costs of each method are roughly the same.

In \cref{fig:wedge}, we present top-$t$ scores and times for wedge and diamond sampling on the Skitter dataset.
In these experiments, we set the budget of dot products to be $t'=1000\cdot t$.
For fewer than $10^5$ samples, the diamond sampling has much better accuracy, but because the time is dominated by preprocessing, diamond sampling is also more expensive.
For greater than $10^5$ samples, the time is dominated by sampling and computing dot products, so the running times of diamond and wedge sampling approach are roughly the same.
However, diamond sampling has identified all top entries by $10^5$ samples, while wedge sampling needs $10^6$ or $10^7$ samples to identify all top entries, requiring an order of magnitude more time.

\begin{figure}[htb]
\centering
\includegraphics[width=.49\columnwidth]{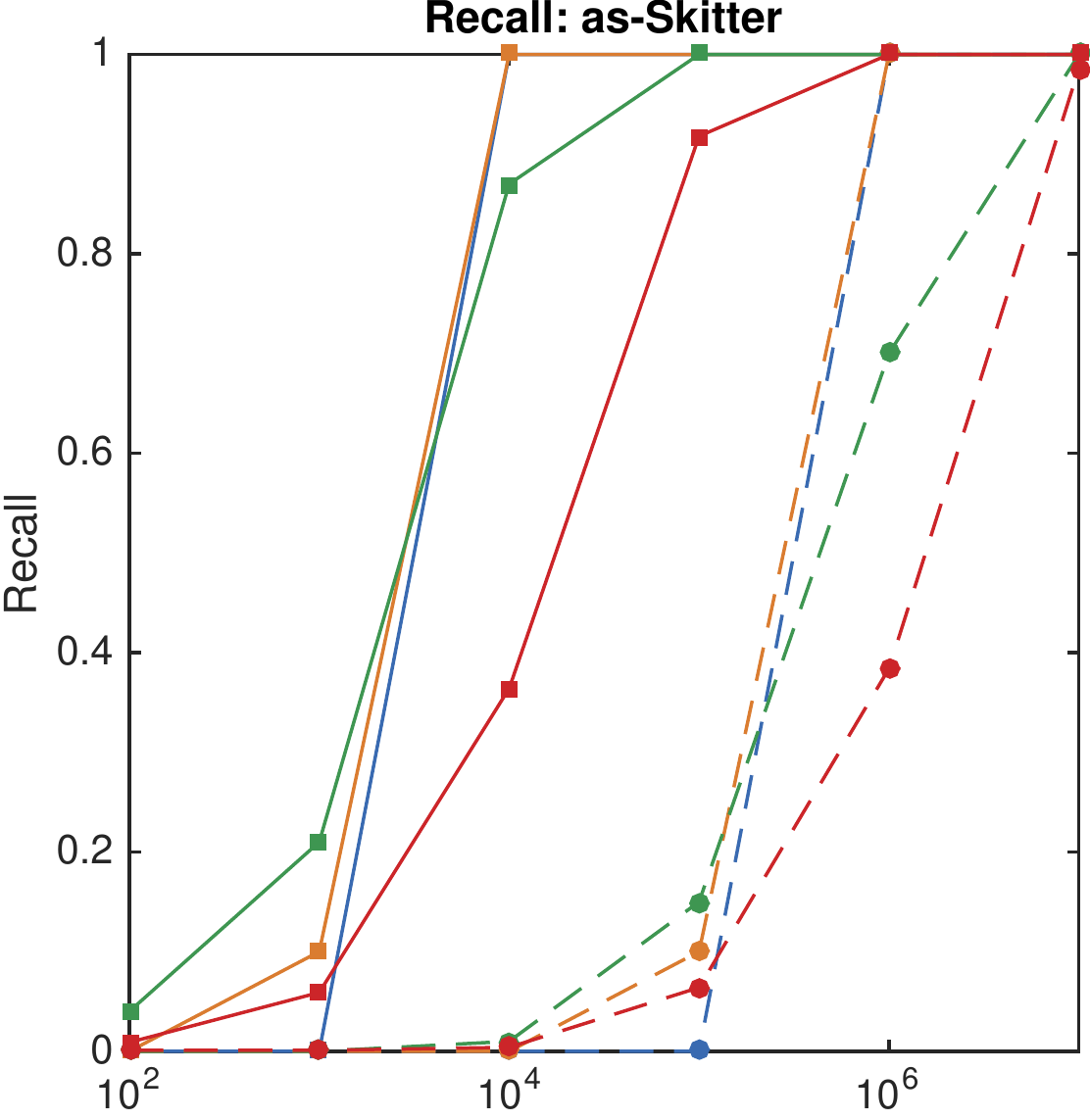}
\includegraphics[width=.49\columnwidth]{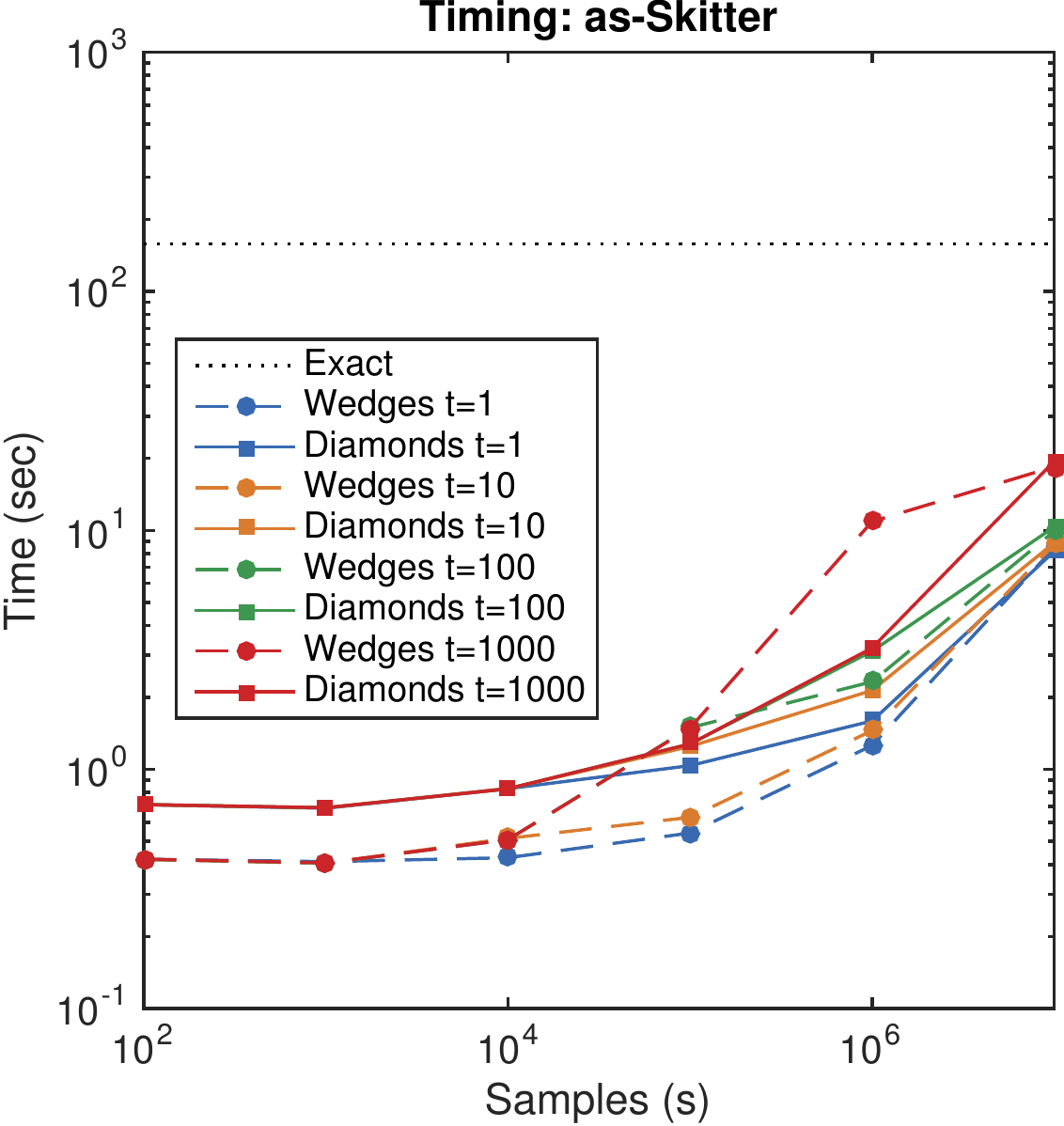}
\caption{Comparison of diamond sampling and wedge sampling.}
\label{fig:wedge}
\end{figure}

\subsection{Comparison to  asymmetric LSH}
\label{sec:comp-asymm-lsh}

We provide a comparison to some experiments in the paper
by Shrivastava and Li \cite{ShLi14} (additional details in
\cite{ShLi14a}) for the Movielens-10M data set \cite{movielens}. 
User $j$ corresponds to column
$\vec{b_j}$. For each user $j$, we want to find the top-10 movie
recommendations.
In other words, for
each user, we want to solve the $k$-MIPS problem.

Precision-recall results over 2000 random users using asymmetric LSH and an
increasing number of hash functions, $h \in \set{64,
128, 256, 512}$, are reported in
\cite{ShLi14,ShLi14a}. 
The amount of storage increases with $h$.
For comparison, we reproduce the curves reports
in their figures, although we did not redo the experiments.

We also pick 2000 random users and apply our diamond sampling approach
in \cref{alg:diamond} repeatedly, using just a \emph{single column} of
$B$ in each application. The $A$ matrix is approximately 79MB in size.
We need to keep one object that is the size of $A$ and a few vectors
of length $m$ or $s$, for a total of less that 100MB extra storage.
While it is hard to pin down the exact storage of LSH methods,
it is on the order of one to two magnitudes more than the dataset.
It is well-known that LSH is memory intensive (this is explicitly called
out in the E2-LSH manual~\cite{LSH-man}).

We use an increasing number of samples,
$s \in \set{64, 128, 256, 512}$.  We set the dot-product budget to be
$t'=s$.  Average precision-recall curves are shown in
\cref{fig:lsh}.
It is difficult to compare the methods directly, so we cannot say that using 64 samples is comparable to using 64 hash functions. However, we can say that the storage per sample is much less than the storage per hash.

\begin{figure}[htbp]
  \centering
  \includegraphics[width=2.5in]{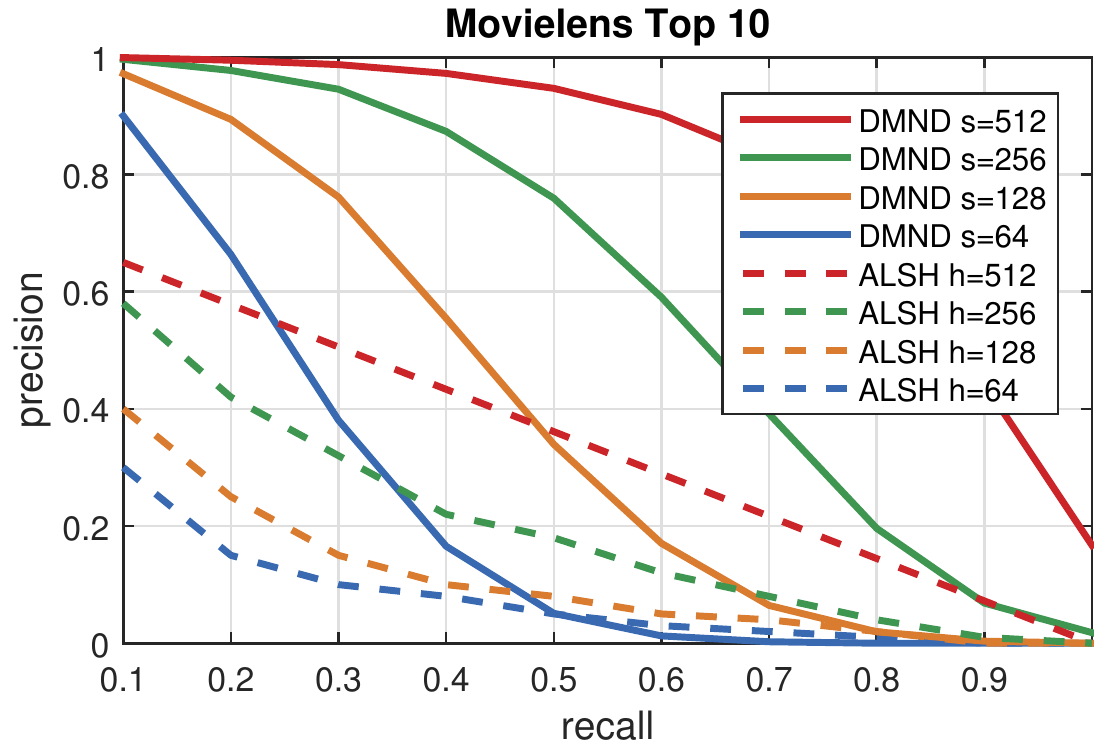}
  \caption{\small Comparison of diamond sampling and asymmetric LSH.}
  \label{fig:lsh}
\end{figure}

 \section*{Acknowledgment}
 Thanks to Madhav Jha and Kevin Matulef for helpful discussions
 regarding this work.
 
 Ballard's work is funded by a Harry S. Truman Postdoctoral Fellowship.
 This material is based upon work supported by the U.S. Department of Energy, Office of Science, Office of Advanced Scientific Computing Research, Complex Interconnected Distributed Systems (CIDS) program and the DARPA GRAPHS program.
 Sandia National Laboratories is a multi-program laboratory managed and operated by Sandia Corporation, a wholly owned subsidiary of Lockheed Martin Corporation, for the U.S. Department of Energy's National Nuclear Security Administration under contract DE-AC04-94AL85000.

\end{document}